\documentclass[journal,twocolumn,10pt]{IEEEtran}
\usepackage{color}
\usepackage{amsfonts}
\usepackage{amsmath}
\usepackage{amsthm}
\usepackage{amssymb}
\usepackage{alltt}
\usepackage{epsfig}
\usepackage{mathrsfs}
\usepackage{graphicx,ifthen}
\usepackage{epstopdf}
\usepackage{subfigure}
\usepackage{esint}
\usepackage{bm}
\usepackage{bbm}
\usepackage{enumerate}
\usepackage{cite}
\usepackage{float}
\usepackage{cuted}
\usepackage{algorithm}
\usepackage{algpseudocode}
\usepackage{booktabs}
\usepackage{array}
\usepackage{ulem}
\usepackage{soul}
\usepackage{verbatim}
\usepackage{hhline}
\usepackage{multirow}
\usepackage{arydshln}
 \usepackage{makecell}
\usepackage{stfloats}
\usepackage{siunitx}

\newtheorem{proposition}{Proposition}

\newcommand{\mb}[1]{\mathbf{#1}}
\newcommand{\mc}[1]{\mathcal{#1}}
\newcommand{\mbb}[1]{\mathbb{#1}}

\begin{document}


\date{} 

{ {	\title{ \LARGE{Analysis and Optimization of Multiple-STAR-RIS Assisted MIMO-NOMA with GSVD Precoding: An Operator-Valued Free Probability Approach}}} }

\author{ {Siqiang~Wang,~\IEEEmembership{Graduate Student Member,~IEEE}, Zhong~Zheng,~\IEEEmembership{Member,~IEEE,} 	Jing~Guo,~\IEEEmembership{Senior Member,~IEEE},	Zesong~Fei,~\IEEEmembership{Senior Member,~IEEE} and Zhi Sun,~\IEEEmembership{Senior Member,~IEEE}
	}
	
	\thanks{S. Wang, Z. Zheng, J. Guo, and Z. Fei are with the School of Information and Electronics, Beijing Institute of Technology, Beijing,
		China (E-mails: \{3120205406, zhong.zheng, jingguo, feizesong \}@bit.edu.cn). 
		
			Z. Sun is with the Department of Electronic Engineering, Tsinghua University, Beijing, China (E-mail: zhisun@tsinghua.edu.cn).
	  }
	
}

\maketitle
\begin{abstract}
Among the key enabling 6G techniques, 
multiple-input multiple-output  
(MIMO) and non-orthogonal multiple-access (NOMA) play an important role in 
 enhancing the spectral efficiency of the wireless communication systems. To further extend the coverage and the capacity, 
  the simultaneously transmitting and reflecting reconfigurable intelligent surface (STAR-RIS) has recently emerged out as a cost-effective technology. 
To exploit the benefit of STAR-RIS in the MIMO-NOMA systems, in this paper, we investigate the analysis and optimization of the downlink dual-user MIMO-NOMA systems assisted by multiple STAR-RISs under the generalized singular value decomposition (GSVD) precoding scheme, in which the channel is assumed to be Rician faded with the Weichselberger’s correlation structure. To analyze the asymptotic information rate of the users,  we apply the operator-valued free probability theory to obtain the Cauchy transform of the generalized singular values (GSVs) of the MIMO-NOMA channel matrices, which can be used to obtain the information rate by Riemann integral. Then, considering the special
case when the channels between the BS and the STAR-RISs are deterministic,  
 we obtain the closed-form expression for the asymptotic information rates of the users.  Furthermore, a projected gradient ascent method (PGAM) is proposed with the derived closed-form expression to design the STAR-RISs thereby maximizing the sum rate based on the statistical channel state information.
 The numerical results show the accuracy of the asymptotic expression compared to the Monte Carlo simulations and the superiority
 of the proposed PGAM algorithm.


\end{abstract}

\begin{IEEEkeywords}
STAR-RIS, MIMO-NOMA, GSVD, Rician channel, operator-valued free probability, PGAM.
\end{IEEEkeywords}

\IEEEpeerreviewmaketitle


\section{Introduction}\label{intro}
Future communication systems are envisioned to support emerging services such as virtual reality, autonomous driving, and augmented reality, which promotes the need to further enhance the spectral efficiency and capacity of the wireless communication systems.
As a key physical layer technology in 5G,
multiple-input multiple-output (MIMO) plays an important role in
improving the spectral efficiency and reliability of the wireless link by utilizing the spatial degree of freedom (DoF) of the {transfer} channels  \cite{SaadNet2020,Liu2023TVT}. However, the capacity of the MIMO system with the orthogonal access scheme is limited when there are many users/streams spatially indistinguishable \cite{DingCM2017}. Therefore, to achieve the capacity gain promised by MIMO with large number of users/streams, the non-orthogonal multiple access (NOMA) schemes have emerged to boost the system capacity \cite{Dai2018CS,Wang2023IOT}, in which the data are transferred in the same time-frequency resource, and then the data are recovered at the receiver by the successive interference cancellation (SIC) mechanism \cite{FangTcom2023,LiuXTVT20242023}.

There are many works on the MIMO-NOMA schemes, among which the generalized singular value decomposition (GSVD) precoding stands out due to the better communication  performance  and low implementation  complexity, especially under  the high signal-to-noise ratio (SNR) \cite{RaoIot2024}. Specifically, 
the GSVD precoding transforms the multi-user MIMO (MU-MIMO) channel into multiple parallel multi-user single-input single-output (MU-SISO) subchannels, thereby  reducing the complexity of signal processing in each subchannel \cite{LoanJNA,MaVTC2016}. {\color{black} Therefore, the performance of the GSVD-precoded MIMO-NOMA systems is gaining increasing attention.}
In \cite{chenTWC2019}, the GSVD precoding scheme was adopted in the downlink MIMO-NOMA system to improve the sum rate and the asymptotic average rates for individual users were derived under Rayleigh channels, assuming that the number of the receiving antennas at the users are the same.
Furthermore, considering the line-of-sight (LoS) components of the channels,
the asymptotic performance of the GSVD-precoded MIMO-NOMA system with Rician fading was investigated in \cite{RAO2023}. With the free probability theory, the authors first derived the Cauchy transform of the generalized singular values (GSVs) of the MIMO-NOMA channels. Then, by simplifying the channels to
be Rayleigh-faded, the closed-form expressions of the average rate were derived. In addition, the authors in  \cite{RaoIot2024} extended the scenario to a GSVD-based MIMO-NOMA system assisted by an amplify-and-forward (AF) relay, {\color{black}in which the closed-form expression of the probability density
function (PDF) of GSVs for two MIMO channel matrices was derived for the insightful performance evaluation.}

Recently, reconfigurable intelligent surface (RIS) is regarded
as a promising technology in 6G to enhance the cellular communication network
since it can  improve the effective channel between the transmitter and the receiver  by controlling the phase and amplitude of the RIS elements \cite{Wu2020CM}. However, conventional RIS can only serve users on one side of the RIS panel, otherwise the signal will be blocked. To address this issue, simultaneously
transmitting and reflecting RIS (STAR-RIS) is proposed to further exploit the merit of RIS to extend the coverage and improve the channel quality \cite{ZuoTWC2023}. 
Specifically, STAR-RIS extends
the service coverage  to the full-space by transmitting and reflecting
the impinging signals simultaneously
 \cite{PapazafeiropoulosTVT2023,XuCL2021}. The distinctive feature of the STAR-RIS enables seamless integration with the NOMA technique,
 i.e., the users located  on both sides of the STAR-RIS can be simultaneously served by the BS with the NOMA scheme. 
 {\color{black} By allocating distinct powers to different users at the BS and configuring different transmission and reflection coefficients at the STAR-RIS,} the signals for the two users can be effectively distinguished.
{\color{black}Therefore, the performance analysis of STAR-RIS assisted NOMA systems has gradually attracted attention \cite{WangCL2022,BasharatTWC2023,Zhaocl2022,ChenWCL2022}.} Specifically, when the direct links between the BS and the users are blocked by the obstacles, the authors in \cite{Zhaocl2022} derived the closed-form expressions of ergodic rates and high signal-to-noise ratio (SNR) slopes for the downlink transmissions towards single-antenna users, where a single STAR-RIS was deployed. Furthermore, considering spatially correlated channels towards two users, {\color{black}the authors in \cite{WangCL2022} derived the closed-form  expressions of the outage probability for two downlink NOMA users with a STAR-RIS assisted}, which illustrated the performance loss due to the channel correlations for a single-input single-output (SISO) system. In addition, considering the existence of the direct links between the BS and the users,  the statistical distribution of the \mbox{Nakagami-m} channels and the ergodic capacity of the receiving node were obtained for 
 a STAR-RIS assisted SISO-NOMA downlink system in \cite{BasharatTWC2023}.  With multiple antennas equipped by BS in \cite{ChenWCL2022}, the authors  analyzed  the approximate
 analytical expressions of the ergodic rate of a STAR-RIS aided downlink NOMA system, which were maximized to design the phase-shift matrix of the STAR-RIS with the statistical channel state information (CSI).

{\color{black}{However, to the best of our knowledge, the area of STAR-RIS assisted MIMO-NOMA system with GSVD precoding remains unexplored and needs to be analyzed to obtain insights  about the impacts of the GSVD-precoding scheme on the STAR-RIS assisted system
		performance.}} In addition, the scenarios considered in the existing works on GSVD-NOMA either have the simplified Rayleigh channel assumptions \cite{chenTWC2019,RaoIot2024}, or derive results that are not general and analytical \cite{RAO2023}. Therefore, due to the more complicated channel configuration of STAR-RIS-assisted communications with general Rician fading, the existing
 performance analysis approaches  cannot be applied here.

{\color{black}To address these issues, we focus on the analysis and optimization of the
	asymptotic information rate of the GSVD-precoded MIMO-NOMA systems assisted by multiple STAR-RIS, which can efficiently improve the coverage and the rank deficiency of the MIMO channels \cite{ZhaoWCL2021,MaTWC2022,PalaWCL2023}. In addition, a general
	Rician channel model with Weichselberger’s correlation
	structure is adopted}. {\color{black} By simplifying the system configuration, the considered system  can degrade to other MIMO-NOMA systems \cite{chenTWC2019,RAO2023}. In addition, compared to \cite{WangTWC2024}, we further develop a linearization method for the rational matrix polynomial and address the inversion of 7x7 block matrix with non-diagonal  R-transform structure.  The correctness of the derived closed-form expression for the information rate is verified  by the derivatives and the information rate relationships under special SNR, which is not concluded in  \cite{WangTWC2024}.}

The main contributions of this paper are summarized as follows:
\begin{itemize}
	\item We consider a multi-STAR-RIS assisted MIMO-NOMA system, where the GSVD-precoding scheme is applied to serve two users simultaneously. In addition, a general Rician
	MIMO channel model with Weichselberger's correlation structure is adopted to fit a wider range of MIMO channels. By applying the operator-valued free probability theory, we obtain the Cauchy transform of the GSVs of the MIMO-NOMA channel matrices under different antenna configurations with the help of the linearization trick for the matrix-valued rational functions.
	\item To obtain a deep insight into the considered system,  we consider a special
	case by setting the channels between the BS and the STAR-RISs deterministic. Then the closed-form expressions for the asymptotic information rates of the two NOMA users are obtained via the Cauchy transform of the MIMO-NOMA channel matrices. {\color{black} In addition, the correctness of the derived closed-form expression for the information rate is proved by verifying derivatives and the rate relationships under special scenarios, i.e.,  the information rate should be 0 as the noise tends to be infinite}
	Furthermore, based on the closed-form expressions, we propose a projected gradient ascent method (PGAM) to jointly optimize the phase shifts, the transmission and reflection coefficients of the STAR-RIS to maximize the sum rate of the two users with statistical CSI. Specifically, with the derived gradients of the sum rate, we iteratively update variables and  project the updated variables onto the appropriate space to satisfy the constant-modulus constraints of STAR-RIS elements, as well as the constraints of the transmission and reflection coefficients.
	\item We conduct comprehensive simulations to verify the accuracy of the derived Cauchy transform and the closed-form expressions for the
	power normalization factor and the  asymptotic information rates, as well as the effectiveness of the proposed PGAM optimization algorithm. It is shown that the derived expressions match the Monte-Carlo simulations well and the STAR-RIS
	can provide significant performance gains  by using the proposed PAGM algorithm.
\end{itemize}

The rest of the paper is organized as follows. The system model is introduced in Section II. In Section III, we derive the Cauchy transform for the GSVs of the STAR-RIS assisted MIMO-NOMA channel and the explicit expression for the GSVD-precoding scheme. The closed-form expressions for the asymptotic information rates and the proposed PGAM algorithm  are demonstrated in Section IV. The simulation results are illustrated in Section V. Finally, we conclude the main findings of the paper in  Section VI.

\emph{Notations}: In this paper, we denote scalars, vectors and matrices as the non-bold letters, lowercase bold letters and uppercase bold letters, respectively. $\bf{0}$ represent all-zero matrix and ${\bf{I}}_N$ represent the $N \times N$ identity matrix. The superscripts $(\cdot)^*$, $(\cdot)^\text{T}$ and $(\cdot)^\dag$ are denoted as the conjugate, transpose and conjugate transpose operations, respectively. The notation $ [\mb{A}]_{i,j} $ is defined as the element in the $i$-th row and $j$-th column of matrix $ \mb{A} $.
The notation  $\text{diag}(\mb{A}_1,\mb{A}_2\cdots,\mb{A}_n)$ denotes the diagonal block matrix consisting of $ \mb{A}_1,\mb{A}_2\cdots,\mb{A}_n $ matrices. The operator $\text{Tr}(\cdot)$ denotes the trace of the matrix, $\mbb{E}[\cdot]$ denotes the expectation, and $\det(\cdot)$ denotes the matrix determinant.

\section{System Model}\label{secModel}

\subsection{Signal Model}
In this paper, we consider a multi-STAR-RIS assisted downlink MIMO-NOMA communication system as shown in Fig.~\ref{simlpifySystem}, where the BS simultaneously transmits individual messages to two users located on different sides of the STAR-RIS panels \footnote{{\color{black}For the scenarios with more than two users, the users can be first grouped in pairs \cite{RAO2023}, in which the users in different groups can be differentiated by occupying orthogonal time-frequency resource blocks (such as orthogonal frequency division multiple access scheme). After grouping the users, the proposed MIMO-NOMA with GSVD precoding scheme for two users can be directly applied.}}. The BS is assumed to be equipped  with $T$ transmit antennas and the $i$-th user is assumed to be equipped with $R_i$ receiving antennas for $i=1,2$.
Let $\mb{s}_i \in \mbb{C}^{T \times 1}$ denote the message sent to the $i$-th user. In the MIMO-NOMA transmission scheme, the transmit signal can be expressed as 
\begin{align}
	\mb{x}=\mb{W}\mb{s}=\mb{W}\left( \sqrt{\kappa_1}\mb{s}_1 +\sqrt{\kappa_2}\mb{s}_2\right),
\end{align}
where $ \mb{W} $ is the precoding matrix at the BS, and $ \kappa_1 $ and $ \kappa_2 $ are the fractional power allocated to the two users, satisfying $ 0\le \kappa_i\le 1 $ and ${\kappa_1}+{\kappa_2}=1$.
In addition, the transmit signals are independent Gaussian such that $ \mb{s}_i\sim CN(\mb{0},\mb{I}) $ and $ \mbb E[\mb{s}_1\mb{s}_{2}^\dag] = \mb{0} $.
Assuming that the noise at the {$i$-th} user $ \mb{n}_i $ is the additive white Gaussian noise (AWGN)  with zero mean and covariance  $\sigma_0^2{\bf{I}}_{R_i}$,
the received signal at the $n$-th user can be expressed as 
\begin{align}
	\mb{y}_i=\sqrt{\frac{{\rho_i}P}{t}}\mb{A}_i\left(\mb{R}_{0,i}+ \sum_{k=1}^{K}\mb{R}_{k,i}\mb{\Theta}_{k,i}\mb{F}_k\right)\mb{W}\mb{s} + \mb{n}_i, \label{receive}
\end{align}
where $ t={\text{Tr}(\mb{W}\mb{W}^\dagger)} $ is the power normalization factor for the precoder,  $P$ is the transmit power of the BS and $ \rho_i $ is the relative channel gain for the $i$-th user.
The matrices  $ \mb{A}_i \in \mbb{C}^{R_i \times R_i}$, $ 1\le i \le 2 $, are the receive filter at the $i$-th user. The diagonal matrices $ \mb{\Theta}_{k,1} =\text{diag}(\sqrt{\beta_{k,1,1}}e^{j\theta_{k,1,1}},...,\sqrt{\beta_{k,1,L_k}}e^{j\theta_{k,1,L_k}})$ and $ \mb{\Theta}_{k,2}=\text{diag}(\sqrt{\beta_{k,2,1}}e^{j\theta_{k,2,1}},...,\sqrt{\beta_{k,2,L_k}}e^{j\theta_{k,2,L_k}})$ denote the transmitting and reflecting phase-shifting matrices of the STAR-RISs, respectively.
The parameters $ 0\le \sqrt{\beta_{k,i,l}} \le 1 $ and $ \theta_{k,i,l} \in [0,2\pi)$ denote the fractional power and the phase shifts  of the transmitting side ($ i=1 $) or the reflecting side ($ i=2 $), where $ \beta_{k,1,l} + \beta_{k,2,l}= 1 $ for $1\leq k\leq K$ and $1\leq l\leq L_k$.
The matrices $\mb{R}_{0,i} \in \mbb{C}^{R_i \times T}$  denote the channels from the BS to the $ i $-th user and the matrices
$\mb{F}_{k} \in \mbb{C}^{L_k \times T}$  denote the channels from the BS to the $k$-th STAR-RIS and $ \mb{R}_{k,i} \in \mbb{C}^{R_i \times L_k} $ denote the channel from the $k$-th STAR-RIS to the $i$-th user, where $ L_k $ is the number of transmitting and reflecting elements in $k$-th STAR-RIS panel. 

\subsection{Channel Model}
\begin{figure}[t]
	\centerline{\includegraphics[width=1\columnwidth]{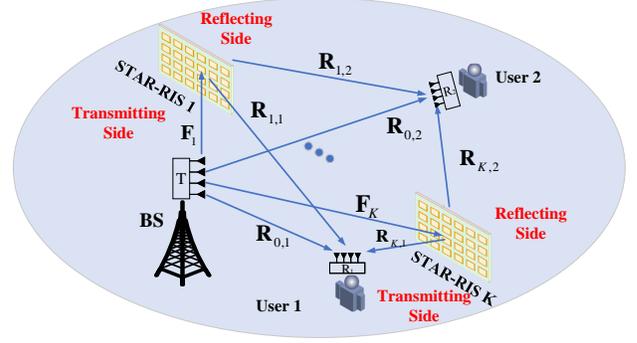}
	}
	\caption{STAR-RIS assisted MIMO-NOMA communication system.}
	\label{simlpifySystem}
\end{figure}

In this paper, to capture a general spatial correlation structure of the channels, 
we assume all the channels  $ \mb{F}_{k} $ and $ \mb{R}_{k,i}  $ are Rician faded with the Weichselberger's correlation structure \cite{WeichselbergerTWC06}, which can be modeled as
\begin{align}
		\mb{F}_{k} & = \overline{\mb{F}}_{k} + \widetilde{\mb{F}}_{k} = \overline{\mb{F}}_{k} + \mb{U}_{k} (\mb{M}_{k} \odot \widetilde{\mb{X}}_{k}) \mb{V}_{k}^{\dagger},\label{eqFk}\\
	\mb{R}_{k,i} &= \overline{\mb{R}}_{k,i} + \widetilde{\mb{R}}_{k,i} = \overline{\mb{R}}_{k,i} + \mb{T}_{k,i} (\mb{N}_{k,i}\odot\widetilde{\mb{Y}}_{k,i}) \mb{S}_{k,i}^{\dagger},\label{eqGk-2}
\end{align} 
where $\mb{F}_{k}$ for $ 1\le k\le K $ and $ \mb{R}_{k,i} $ for $ 1\le i\le 2, 0\le k\le K$. The matrices $ \overline{\mb{F}}_{k}  $ and $ \overline{\mb{R}}_{k,i} $ are the deterministic LoS components of the channels. The deterministic matrices $  \mb{U}_{k}  $, $ \mb{V}_{k} $, $ \mb{T}_{k,i} $, and  $ \mb{S}_{k,i} $ are unitary, which denote the correlation between the channel coefficients. The elements in the matrices $\mb{M}_{k}$ and $ \mb{N}_{k,i} $ are nonnegative that control the variance of $ \widetilde{\mb{F}}_{k}  $ and $ \widetilde{\mb{R}}_{k,i} $, respectively. The matrices $ \widetilde{\mb{X}}_{k} $ and $ \widetilde{\mb{Y}}_{k,i} $ are complex Gaussian random matrices with  independent identically distributed variables with zero mean and equal variance, i.e., $[\widetilde{\mb{X}}_{k}]_{l,j} \sim\mc{CN}(0,1/T)$,  $[\widetilde{\mb{Y}}_{0,i}]_{l,j} \sim\mc{CN}(0,1/R_i)$ and $[\widetilde{\mb{Y}}_{k,i}]_{l,j} \sim\mc{CN}(0,1/L_k)$ for $1\leq k\leq K$.

In addition, the channels $\mb{F}_{k}$ and $ \mb{R}_{k,i} $  corresponding to different links are assumed to be independent. To facilitate subsequent derivations, we define the parameterized one-sided correlation matrices of $\widetilde{\mb{R}}_{k,i}$ as 
\begin{align}
	\eta_{k,i}(\widetilde{\mb{C}}) &= \mbb{E}[\widetilde{\mb{R}}_{k,i}^\dagger\widetilde{\mb{C}} \widetilde{\mb{R}}_{k,i}] = \frac{1}{\hat{L}_{k,i}} {\mb{S}}_{k,i}\mb{\Pi}_{k,i}(\widetilde{\mb{C}}){\mb{S}}_{k,i}^\dagger,   \label{yta}
	\\
	\widetilde{\eta}_{k,i}({\mb{C}}) &= \mbb{E}[\widetilde{\mb{R}}_{k,i}{\mb{C}} \widetilde{\mb{R}}_{k,i}^\dagger] = \frac{1}{\hat{L}_{k,i}}\mb{T}_{k,i} \widetilde{\mb{\Pi}}_{k,i}({\mb{C}}) \mb{T}_{k,i}^\dagger, \label{yta_qta}
\end{align}
where $ 0 \leq k\leq K $, $ 1\le i\le 2 $ and $ \hat{L}_{k,i}= R_i$ for $k=0$ and  $ \hat{L}_{k,i}= L_k$ for $k>0$.
The matrices $ \widetilde{\mb{C}} \in \mbb{C}^{R_i\times R_i}$ and $ {\mb{C}} \in \mbb{C}^{\hat{L}_{k,i} \times \hat{L}_{k,i}}$ are arbitrary Hermitian matrices. The entries of the diagonal matrices $\mb{\Pi}_{k,i}(\widetilde{\mb{C}})$ and $\widetilde{\mb{\Pi}}_{k,i}({\mb{C}})$ are  
\begin{align}
	\left[\mb{\Pi}_{k,i}(\widetilde{\mb{C}})\right]_{l,l} \!&= \sum_{j = 1}^{R_i} \left([\mb{N}_{k,i}]_{j,l}\right)^2 \left[\mb{T}_{k,i}^\dagger \widetilde{\mb{C}} \mb{T}_{k,i}\right]_{j,j},  1\le l\le \hat{L}_{k,i},
	\\
	\left[\widetilde{\mb{\Pi}}_{k,i}({\mb{C}})\right]_{l,l} &= \sum_{j=1}^{\hat{L}_{k,i}} \left([\mb{N}_{k,i}]_{l,j}\right)^2 \left[ {\mb{S}}_{k,i}^\dagger {\mb{C}} {\mb{S}}_{k,i} \right]_{j,j},   1\le l\le R_i.
\end{align}
Similarly, the parameterized one-sided correlation matrices of $ \widetilde{\mb{F}}_{k} $ are defined as
\begin{align}
		\zeta_{k}(\mb{D}_{1}) &= \mbb{E}[\widetilde{\mb{F}}_{k}^\dagger\mb{D}\widetilde{\mb{F}}_{k,i}] = \frac{1}{T} \mb{V}_{k}\mb{\Sigma}_{k}(\mb{D})\mb{V}_{k}^\dagger, \label{zeta-k-n}
	\\
	\widetilde{\zeta}_{k}(\widetilde{\mb{D}}) &= \mbb{E}[\widetilde{\mb{F}}_{k}\widetilde{\mb{D}}\widetilde{\mb{F}}_{k}^\dagger] = \frac{1}{T}\mb{U}_{k} \widetilde{\mb{\Sigma}}_{k}(\widetilde{\mb{D}}) \mb{U}_{k}^\dagger, \label{zeta_qta_k-n}
\end{align}
where $1 \leq k \leq K$ and the matrices  $ {\mb{D}} \in \mbb{C}^{L_k\times L_k}$ and $ \widetilde{\mb{D}} \in \mbb{C}^{T \times T}$ are arbitrary Hermitian matrices. The diagonal matrices $\mb{\Sigma}_{k}$ and $\widetilde{\mb{\Sigma}}_{k}$ contain the diagonal entries as
\begin{align}
	\left[\mb{\Sigma}_{k}(\mb{D})\right]_{l,l} &= \sum_{j = 1}^{L_k} \left([\mb{M}_{k}]_{j,l}\right)^2 \left[\mb{U}_{k}^\dagger \mb{D} \mb{U}_{k}\right]_{j,j},\quad 1\le l\le T,\\
	\left[\widetilde{\mb{\Sigma}}_{k}(\widetilde{\mb{D}})\right]_{l,l} &= \sum_{j=1}^{T} \left([\mb{M}_{k}]_{l,j}\right)^2 \left[ \mb{V}_{k}^\dagger \widetilde{\mb{D}} \mb{V}_{k} \right]_{j,j}, \quad 1\le l\le L_k.\label{zeta-k-n-center}
\end{align}

\subsection{Information Rate for GSVD Precoding}
In this paper, the  GSVD precoding and decoding scheme is applied to  eliminate the  interference between  the MIMO-NOMA subchannels. Defining $\mb{H}_i=\left(\mb{R}_{0,i}+ \sum_{k=1}^{K}\mb{R}_{k,i}\mb{\Theta}_{k,i}\mb{F}_k\right)$, we can obtain the GSVD precoding matrix  $\mb{W}$ and the decoding matrices $\mb{A}_i$ for two users  as $ \mb{W}=\mb{V}^{-1} $ and 
$ \mb{A}_i=\hat{\mb{U}}_i^\dagger $, where $ \mb{V} $ and $ \hat{\mb{U}}_i $ are obtained by applying GSVD to the channel matrices $ \mb{H}_1 $ and $ \mb{H}_2 $, which can be expressed as 
\begin{align}
	\mb{H}_i=\hat{\mb{U}}_i\mb{\Sigma}_i\mb{V},
\end{align}
where $ \hat{\mb{U}}_i $  is a unitary matrix and $ \mb{V} $ is  a nonsingular matrix and $ \mb{\Sigma}_i $ is a rectangular diagonal matrix. With the GSVD precoding matrix  $\mb{W}$ and the decoding matrices $\mb{A}_i$, the received signal in \eqref{receive} can be repressed as 
{\color{black}\begin{align}
	\mb{y}_i=\sqrt{\frac{\rho_iP}{t}}\mb{A}_i\mb{H}_i\mb{W}\mb{s} + \mb{n}_i=\sqrt{\frac{\rho_iP}{t}}\mb{\Sigma}_i\mb{s} + \mb{n}_i.
\end{align}}

Therefore, the expressions for information rate are related to the matrix $\mb{\Sigma}_i$, which is dependent on the numbers of transmitting antennas and receiving antennas.  {\color{black}With the statistical CSI obtained, we adopt a statistical CSI-based SIC scheme to eliminate the interference between two users \cite{RAO2023,WangWCL2018}, in which the user with better average channel conditions conducts SIC, i.e., we assume $ \mbb{E}\left[\text{Tr }\left(\rho_1 \mb{H}_1\mb{H}_1^\dagger\right)\right]>\mbb{E}\left[\text{Tr }\left(\rho_2 \mb{H}_2\mb{H}_2^\dagger\right)\right] $ and the SIC is performed at user 1 while user 2 treats the signal of user 1 as noise.} Let $ \hat{m}=\min(R_1,R_2)$,
according to the definition of GSVD, the information rates for two users are summarized as follows:
\begin{itemize}
\item When $T \leq \hat{m}$: The rectangular diagonal matrix $\mb{\Sigma}_i \in \mbb{C}^{R_i \times T}$ can be expressed as $\mb{\Sigma}_i=\left[\begin{array}{*{20}{c}}
	{\hat{\mb{\Sigma}}_i}\\
	{\mb{0}_{\hat{R}_i\times T}}
\end{array}\right]$, where $\hat{R}_i=R_i-S$, $\hat{\mb{\Sigma}}_i=\text{diag}({\sigma_{i,1},...,\sigma_{i,S}})$ and $S=\min(R_1,T)+\min(R_2,T)-\min(R_1+R_2,T) $.  Then the information rate for the users can be expressed as \begin{align}
	I_{1}=&\sum_{k=1}^{S}\log \left( 1+  \frac{\mu_k}{1+\mu_k}\frac{P\kappa_1\rho_1}{t\sigma_0^2} \right), \label{I1_ori}\\
	I_{2}=&\sum_{k=1}^{S}\log \left( 1+  \frac{\kappa_2P}{\kappa_1+(1+\mu_k)t\sigma_0^2} \right), \label{I2_ori}
\end{align}
where $\mu_k=\frac{\sigma_{1,k}^2}{\sigma_{2,k}^2}$.

\item  When  $\hat{m} \leq T < R_1+ R_2$: The rectangular diagonal matrix $\mb{\Sigma}_i $ can be expressed as
\begin{align}
\mb{\Sigma}_1=	\left[ {\begin{array}{*{20}{c}}
			{\mb{0}_{S\times \hat{T}_1 }}&{\hat{\mb{\Sigma}}_1}&{\mb{0}_{S\times\hat{R}_1}}\\
			{\mb{0}_{\hat{R}_1\times \hat{T}_1 }}&{\mb{0}_{\hat{R}_1\times S }}&{ \mb{I}_{\hat{R}_1}}
	\end{array}} \right],\\
\mb{\Sigma}_2=\left[ {\begin{array}{*{20}{c}}
	{\mb{I}_{\hat{R}_2 }}&{\mb{0}_{\hat{R}_2\times S }}&{\mb{0}_{\hat{R}_2 \times\hat{T}_2}}\\
	{\mb{0}_{S\times\hat{R}_2 }}&	{\hat{\mb{\Sigma}}_2}&{ \mb{0}_{S\times \hat{T}_2}}
	\end{array}} \right].
\end{align}
Then the information rates for the users can be expressed as \begin{align}
	\hat{I}_{1}=&I_{1}+ (T-R_1)\log\left(1+ \frac{P\rho_1}{t\sigma_0^2}\right), \\
	\hat{I}_{2}=&I_{2}+(T-R_2)\log\left(1+ \frac{P\rho_2}{t\sigma_0^2}\right), 
\end{align}
where $I_1$ and $I_2$ are given in \eqref{I1_ori} and \eqref{I2_ori}.

\item When $T \geq   R_1+ R_2$: $\mb{\Sigma}_1=\left[\mb{I}_{R_1} \mb{0}_{R_1 \times (T-R_1)}\right] $, $\mb{\Sigma}_2=\left[\mb{0}_{R_2 \times (T-R_2)} \mb{I}_{R_2}\right] $ and 
the information rates for the users can be expressed as \begin{align}
	\hat{I}_{1}=&(T-R_1)\log\left(1+ \frac{P\rho_1}{t\sigma_0^2}\right), \\
	\hat{I}_{2}=&(T-R_2)\log\left(1+ \frac{P\rho_2}{t\sigma_0^2}\right). 
\end{align}
\end{itemize}

Based on the above analysis, we can observe that only $I_1$ and $I_2$ in \eqref{I1_ori} and \eqref{I2_ori} are related to the channel. Therefore, we only focus on the analysis for $I_1$ and $I_2$. However,  it is difficult to directly obtain the expression of $ I_i $  since the distribution of $\mu_k$ is unknown. Therefore, we rewrite $I_1$ and $I_2$ equivalently in term of the Cauchy transform of $ \mc{G}_\mu $ as \cite{chenTWC2019}
\begin{align}
	I_1=&S\int_{0}^{\frac{\rho_1}{t\sigma_0^2}} { \left( \frac{1}{1+x}+  \frac{1}{(1+x)^2}\mc{G}_\mu(-(x+1)^{-1}) \right)\text{d}x}, \label{R1_int}\\
	I_2=&S\int_{0}^{\frac{\rho_2}{t\sigma_0^2}} {\left( -\mc{G}_\mu(-(x+1))+\kappa_1\mc{G}_\mu(-(\kappa_1x+1)) \right)\text{d}x},\label{R2_int}
\end{align}
where $ \mc{G}_\mu(z)= \int_{0}^{\infty} \frac{1}{z-x}\mathrm{d}F_{\mu}(x)$ and $F_{\mu}$ denotes the empirical cumulative distribution function of $\mu$.
According to \cite{chenTWC2019}, when $ R_1>R_2 $, the variable $\mu$ are the non-zero eigenvalues of the matrix $ \mb{B} $ as
\begin{align}
	\mb{B}=
	\left\{ {\begin{array}{*{20}{c}}
			{\mb{H}_1(\mb{H}_2^\dagger\mb{H}_2)^{-1}\mb{H}_1^\dagger,}&{T\leq R_2 }\\
			{{\mb{H}}_1(\hat{\mb{H}}_2^\dagger\hat{\mb{H}}_2)^{-1}{\mb{H}_1^\dagger},}&{R_2 < T < R_1+R_2}
	\end{array}} \right. \label{B}
\end{align}
where $ \hat{\mb{H}}_2=\left[ {\begin{array}{*{20}{c}}
		{\mb{H}_2}\\
		{\Delta \hat{\mb{I}} } 
\end{array}} \right] $ and $\hat{\mb{I}}=[\mb{I}_{\hat{R}_2} \mb{0}_{\hat{R}_2\times T } ]$, in which $\Delta \to 0 $ \cite{RAO2023}.  When $ R_2 > R_1 $, $ \mu $ corresponds to the eigenvalues of $ \mb{B} $ by switching $ \mb{H}_1 $ and $ \mb{H}_2 $ in \eqref{B}.
Then the Cauchy transform $ \mc{G}_\mu(z) $ for $ R_2 \leq R_1 $ can be expressed as 
\begin{align}
	\mc{G}_\mu(z)=\left\{ {\begin{array}{*{20}{c}}
			{\frac{R_1}{S}\mc{G}_\mb{B}(z)-\frac{R_1-S}{Sz},}&{T\leq R_2 }\\
			{\frac{R_1}{S}\mc{G}_\mb{B}(z),}&{R_2 < T < R_1+R_2}
	\end{array}} \right. \label{Gu}
\end{align}
where $ \mc{G}_\mb{B}(z) $ denotes the  Cauchy transform of $ \mb{B} $. In the next section, we will focus on the derivation of $ \mc{G}_\mb{B}(z) $ and the expression of the normalization factor $t$, which  will be  used to derive the closed-form expression of the information rates of users in \eqref{R1_int} and \eqref{R2_int}.

\section{Cauchy Transform and the Power Normalization Factor of the STAR-RIS Assisted MIMO-NOMA with GSVD Precoding}
 \begin{figure*}[b!] 
	\hrulefill
	{
		\begin{align}
			\mbb{E}_\mc{D}\left[\mb{X}\right] = \left[\begin{array}{c:c:c:c:c:c:c}
				\mbb{E}\left[\{\mb{X}\}_1^{R_1}\right] & & & & & &\\\hdashline
				& \mbb{E}\left[\mb{X}_{\widetilde{\mb{D}}_1}\right] & & & &\mbb{E}\left[\mb{X}_{\widetilde{\mb{D}}_{3}}\right]&\\\hdashline
				& & \mbb{E}\left[\{\mb{X}\}_{2R_1+L+1}^{2R_1+L+T}\right] & & &&\\ \hdashline
				& & &  \mbb{E}[\mb{X}_{\widetilde{\mb{C}}_1}]& && \\ \hdashline 
				& &  & &\mbb{E}\left[\{\mb{X}\}_{2R_1+R_2+2L+T+1}^{2R_1+R_2+2L+T+R_2}\right] &&\\ \hdashline
				&  \mbb{E}\left[\mb{X}_{\widetilde{\mb{D}}_{4}}\right]& & & & \mbb{E}\left[\mb{X}_{\widetilde{\mb{D}}_2}\right] &\\ \hdashline
				& & & & &&\mbb{E}[\mb{X}_{\widetilde{\mb{C}}_2}]
			\end{array}\right],\label{eqEDX}
		\end{align}
	}
	
\end{figure*}

 {\color{black}Since the information rates $ I_1 $ and $ I_2 $ have an explicit relationship with the Cauchy transform as shown in (23)-(24), in this section we will first focus on the  derivations of the Cauchy transform of $ \mb{B} $.  Then the asymptotic expression for the  power normalization factor $t$ is obtained by the similar way. It should be noted that the derivation for the case $ T> R_2 $ is similar to that for $ T\leq R_2 $ by substituting $ \hat{\mb{H}}_2 $ into $ {\mb{H}}_2 $ in the derivations.}
 
 Free probability theory is a powerful tool to derive the spectral
 distributions of the non-commutative random variables, which are free on a non-commutative probability space \cite{Muller2002TIT}. Therefore, the Cauchy transform of the products of asymptotically free variables can be obtained by free multiplicative convolution in the conventional free probability theory.
 {\color{black} However, in the considered problem with $ \mb{B}=\mb{H}_1(\mb{H}_2^\dagger\mb{H}_2)^{-1}\mb{H}_1^\dagger $, both $ \mb{H}_1 $ and $ \mb{H}_2 $ are  composed of random non-central matrices ($ \mb{F}_k $ and $ \mb{R}_{k,i} $) with non-trivial spatial
 	correlations, and thus, are not free over the complex algebra  in the classic free probability aspect, which makes it difficult to obtain the Cauchy transform.  Therefore, in order to get the Cauchy transform, we have to find a space where $ \mb{H}_1 $ and $ \mb{H}_2 $ are asymptotically free. With this motivation,  we resort to the operator-valued free probability theory and the  linearization trick for the matrix-valued rational function $\mb{B}$, which constructs a corresponding block matrix by embedding the component matrices within $\mb{H}_1$ and $\mb{H}_2$. In this way, the constructed  operator-valued variables with the linearization trick  are shown to be asymptotically free in the operator-valued probability space and then the Cauchy transform of $ \mb{B} $ can be
 	obtained by its relation with the operator-valued Cauchy
 	transform for the linearized block matrix.}
 For notational
 simplicity, we define $\mb{G}_i=[\mb{I}_R, \mb{R}_{1,i}\mb{\Theta}_{1,i},...,\mb{R}_{K,i}\mb{\Theta}_{K,i}]$ and $\mb{F}_i=[\mb{R}_{0,i}^\dagger, \mb{F}_{1}^\dagger,...,\mb{F}_{K}^\dagger]^\dagger$ for $1\leq i\leq 2$. In addition, $\overline{\mb{G}}_i=[\mb{I}_R, \overline{\mb{R}}_{1,i}\mb{\Theta}_{1,i},...,\overline{\mb{G}}_{K,i}\mb{\Theta}_{K,i}]$ and $\overline{\mb{F}}_i=[\overline{\mb{F}}_{0,i}^\dagger, \overline{\mb{F}}_{1}^\dagger,...,\overline{\mb{F}}_{K}^\dagger]^\dagger$ for $1\leq i\leq 2$.
 
\subsection{Cauchy Transform $ \mc{G}_\mb{B}(z) $} 
{\color{black} Based on \eqref{R1_int}-\eqref{Gu}, we need to obtain Cauchy transform $ \mc{G}_\mb{B}(z) $ to get the information rate for two users.}
  According to \eqref{B}, the matrix $ \mb{B} $ takes the
 form of  a rational function of $ \mb{H}_1  $ and $ \mb{H}_2 $.
Applying the linearization trick for the matrix-valued rational functions \cite{HELTON}, we can construct the linearization matrix of $ \mb{B} $ as
 \begin{align}
 		\setlength{\arraycolsep}{1.8pt}
 	{\mb{L}}=\left[ {\begin{array}{*{20}{c}}
 			\mb{0}&	\mb{0}&	\mb{0}&	\mb{0}&	\mb{0}&	\mb{0}&	\mb{G}_1\\
 			\mb{0}&	\mb{0}&	\mb{F}_1&	\mb{0}&	\mb{0}&	\mb{0}&-\mb{I}_{R_1+L}\\
 			\mb{0}&	\mb{F}_1^\dagger&	\mb{0}&	\mb{0}& \mb{0}&-\mb{F}_2^\dagger&\mb{0}\\
 			\mb{0}&	\mb{0}&	\mb{0}&	\mb{0}&-\mb{G}_2^\dagger&\mb{I}_{R_2+L}&	\mb{0}\\
 			\mb{0}&	\mb{0}&	\mb{0}&-\mb{G}_2&\mb{I}_{R_2}&	\mb{0}&	\mb{0}\\
 				\mb{0}&	\mb{0}&-\mb{F}_2&	\mb{I}_{R_2+L}&\mb{0}&\mb{0}&	\mb{0}\\
 					\mb{G}_1^\dagger&	-\mb{I}_{R_1+L}&	\mb{0}&	\mb{0}&\mb{0}&\mb{0}&	\mb{0}\\
 	\end{array}} \right],
 \end{align}
where $L= \sum_{k=1}^{K}L_k$.

Let algebra $ D $ denote as the $ n \times  n $ block matrix and $ n=3R_1+3R_2+4L+T $,
then according to the operator-valued free probability theory,
 the Cauchy transform of $ \mb{B} $ is related to the operator-valued Cauchy transform of $ \mb{L} $ as
  \begin{align}
  	\mc{G}_\mb{B}(z) = \frac{1}{R_1}\mathrm{Tr}\left(\left\{\mc{G}_{\mb{L}}^{\mc{D}}(\mb{\Lambda}(z))\right\}^{(1,1)}\right),\label{eqGB_GL}
  \end{align}
  where $\{\cdot\}^{(1,1)}$ denotes the upper-left $R_1\times R_1$ matrix block and the operator-valued Cauchy transform $ \mc{G}_{\mb{L}}^{\mc{D}}(\mb{\Lambda}(z))  $ is defined as 
  \begin{align}
  \mc{G}_{\mb{L}}^{\mc{D}}(\mb{\Lambda}(z)) = \mathrm{id}\circ\mbb{E}_{\mc{D}}\left[\left(\mb{\Lambda}(z) - \mb{L}\right)^{-1}\right],
  \end{align}
 where $\mathrm{id}$ denotes the identity operator on a Hilbert space and
  the diagonal matrix $\mb{\Lambda}(z)$ is defined as
  \begin{align}
  	\mb{\Lambda}(z) = \begin{bmatrix}
  		z\mb{I}_{R_1} & \mb{0}_{R_1\times (n-R_1)}\\
  		\mb{0}_{(n-R_1)\times R_1} & \mb{0}_{(n-R_1)\times (n-R_1)} 
  	\end{bmatrix}.\label{eqLambdaz}
  \end{align}
The notation $\mbb{E}_\mc{D}\left[\mb{X}\right]$ is defined as
the expectation of  the $n \times n$ block matrix of $ \mb{X} $, which can be expressed as \eqref{eqEDX} at the bottom of this page, in which the notation $ { \{ \mb{A} \} ^{b}_{a}}    $ means extracting the submatrix of $\bf A$ with elements of the rows and columns with indices from $ a $ to $ b $, i.e.,  $  \left[\{ \mb{A} \} ^{b}_{a}\right]_{i,j} = \left[ \mb{A}\right]_{i+a-1,j+a-1}  $ for $ 1\le i,j\le b-a+1$.
The expectations $ \mbb{E}[\mb{X}_{\widetilde{\mb{D}}_1}] $, $ \mbb{E}[\mb{X}_{\widetilde{\mb{D}}_2}] $, $ \mbb{E}[\mb{X}_{\widetilde{\mb{C}}_1}] $ and $ \mbb{E}[\mb{X}_{\widetilde{\mb{C}}_2}] $ are defined as
\begin{align}
\mbb{E}[\mb{X}_{\widetilde{\mb{D}}_1}] &=\text{diag}\{  \mbb{E}\left[\{\mb{X}\}_{R_1+1}^{2R_1}  \right],\mbb{E}\left[\{\mb{X}\}_{2R+1}^{2R+L_1}\right], \notag \\
	& \quad \quad \quad \quad \quad \quad  \cdots, 
	\mbb{E}\left[\{\mb{X}\}_{2R_1+\sum_{k=1}^{K-1}L_k+1 }^{2R_1+\sum_{k=1}^{K}L_k}\right]
	\},\\
	\mbb{E}[\mb{X}_{\widetilde{\mb{D}}_2}] &=\text{diag}\{  \mbb{E}\left[\{\mb{X}\}_{2R_1+2R_2+2L+T+1}^{2R_1+3R_2+2L+T}  \right],\mbb{E}\left[\{\mb{X}\}_{D+1}^{D+L_1}\right],  \notag \\ &
 \quad \quad \quad \quad \quad \quad  \cdots,
	\mbb{E}\left[\{\mb{X}\}_{D+\sum_{k=1}^{K-1}L_k+1 }^{D+\sum_{k=1}^{K}L_k}\right]
	\},\\
	 \mbb{E}[\mb{X}_{\widetilde{\mb{C}}_1}] &=\text{diag}\{  \mbb{E}\left[\{\mb{X}\}_{2R_1+L+T+1}^{2R_1+R_2+L+T}  \right],\mbb{E}\left[\{\mb{X}\}_{C_1+1}^{C_1+L_1}\right], \notag \\
	& \quad \quad \quad \quad \quad \quad  \cdots,   
	\mbb{E}\left[\{\mb{X}\}_{C_1+\sum_{k=1}^{K-1}L_k+1 }^{C_1+\sum_{k=1}^{K}L_k}\right]
	\}. \\
		 \mbb{E}[\mb{X}_{\widetilde{\mb{C}}_2}]& =\! \text{diag}\{  \mbb{E}\left[\{\mb{X}\}_{2R_1+3R_2+3L+T+1}^{3R_1+3R_2+3L+T}  \right],\mbb{E}\left[\{\mb{X}\}_{C_2+1}^{C_2+L_1}\right], \notag \\
	& \quad \quad \quad \quad \quad \quad   \cdots, 
	\mbb{E}\left[\{\mb{X}\}_{C_2+\sum_{k=1}^{K-1}L_k+1 }^{C_2+\sum_{k=1}^{K}L_k}\right]
	\},\\
	 \mbb{E}[\mb{X}_{\widetilde{\mb{D}}_{3}}]& =\! \text{diag}\{  \mbb{E}\left[\{\mb{X}\}_{R_1+1:2R_1}^{D-R_2+1:D}  \right],\mbb{E}\left[\{\mb{X}\}_{2R_1+1:2R_1+L_1}^{D+1:D+L_1}\right], \notag \\
	& \quad \quad \quad    \cdots, 
	\mbb{E}\left[\{\mb{X}\}_{2R_1+\sum_{k=1}^{K-1}L_k+1:2R_1+L }^{D+\sum_{k=1}^{K-1}L_k+1:D+L}\right]
	\}, \\
		 \mbb{E}[\mb{X}_{\widetilde{\mb{D}}_{4}}]& =\! \text{diag}\{  \mbb{E}\left[\{\mb{X}\}_{D-R_2+1:D}^{R_1+1:2R_1}  \right],\mbb{E}\left[\{\mb{X}\}_{D+1:D+L_1}^{2R_1+1:2R_1+L_1}\right], \notag \\
	& \quad \quad \quad    \cdots, 
	\mbb{E}\left[\{\mb{X}\}_{D+\sum_{k=1}^{K-1}L_k+1:D+L2 }^{R_1+\sum_{k=1}^{K-1}L_k+1:2R_1+L}\right]
	\},
\end{align}
where $D=2R_1+3R_2+2L+T$, $C_1=2R_1+R_2+L+T$ and $C_2=3R_1+3R_2+3L+T$. The notation $ { \{ \mb{A} \} ^{a:b}_{c:d}}    $ means extracting the submatrix of $\bf A$ with elements of the rows with indices from $ c $ to $ d $ and columns with indices from $ a $ to $ b $.

Consider the large-dimensional regime as 
 \begin{align}
 T \to \infty, L\to \infty,R_i\to \infty, L/T=\epsilon_1, R_i/T=\epsilon_2, \label{infty}
 \end{align}
where $ \epsilon_1 $ and $ \epsilon_2 $ are constants,
 then based on the operator-valued free probability theory, we can obtain the Cauchy transform $ \mc{G}_\mb{B}(z) $ in the following proposition:
{\color{black}
	\begin{proposition}\label{prop_cauchy1}
	The Cauchy transform of $\mb{B}$, with $z\in\mbb{C}^+$ and \eqref{infty} holding, is given by
	\begin{align}\label{eqGB}
		\mc{G}_{\mb{B}}(z) = \frac{1}{R_1}\mathrm{Tr}\left[	\mc{G}_{1}(z)\right],
	\end{align}
	where 	$ \mc{G}_{1}(z) $ is defined as 
	\begin{align}
		\mc{G}_{1}(z)	=\left(\widetilde{\mb{\Psi}}- \overline{\mb{G}}_1\left( -\widetilde{\mb{\Phi}}_{1,1}- \mb{A}_1^{-1}\right)\overline{\mb{G}}_1^\dagger \right),
	\end{align}
which	is	obtained by solving the fixed point matrix-valued equations iteratively,  shown as  equations \eqref{A1}-\eqref{G9} at the top of the next page,
 \begin{figure*}[t] 
	{
		\begin{align}
& \mb{A}_1	=\left(-\widetilde{\mb{\Psi}}_{1,1}-\overline{\mb{F}}_1\left( -\mb{\mho}-\overline{\mb{F}}_2^\dagger\mb{\Pi}^{-1}\overline{\mb{F}}_2 \right)^{-1}\overline{\mb{F}}_1^\dagger +\widetilde{\mb{\Psi}}_{1,2}\left(-\widetilde{\mb{\Psi}}_{2,2}+\overline{\mb{F}}_2\mb{\mho}^{-1}\overline{\mb{F}}_2^\dagger-\mb{\Omega}^{-1}  \right)^{-1}\mb{\Omega}^{-1}\overline{\mb{F}}_1^\dagger  \right.\notag \\
& \quad  \quad \quad  \left. + \overline{\mb{F}}_1\mb{\Omega}^{-1}\left(-\widetilde{\mb{\Psi}}_{2,2}+\overline{\mb{F}}_2\mb{\mho}^{-1}\overline{\mb{F}}_2^\dagger  \right)^{-1} \widetilde{\mb{\Psi}}_{2,1}  \widetilde{\mb{\Psi}}_{1,2}\overline{\mb{F}}_1\left(-\widetilde{\mb{\Psi}}_{2,2}+\overline{\mb{F}}_2\mb{\mho}^{-1}\overline{\mb{F}}_2^\dagger-\mb{\Omega}^{-1}  \right)^{-1}\widetilde{\mb{\Psi}}_{2,1} \right)^{-1}, \label{A1}	\\
&\mc{G}_{2}(z)=\left( -\widetilde{\mb{\Psi}}_{1,1}-\mb{\Sigma}^{-1}-\overline{\mb{F}}_1\mb{A}_2\overline{\mb{F}}_1^\dagger +\widetilde{\mb{\Psi}}_{1,2}\mb{A}_3\overline{\mb{F}}_1^\dagger+\overline{\mb{F}}_1\mb{A}_4\widetilde{\mb{\Psi}}_{2,1}-\widetilde{\mb{\Psi}}_{1,2}\mb{A}_5\widetilde{\mb{\Psi}}_{2,1} \right)^{-1},\\
&\mc{G}_{3}(z)=\left( -\mho-\overline{\mb{F}}_1^\dagger\mb{O}^{-1}\overline{\mb{F}}_1-\left(\overline{\mb{F}}_2^\dagger+\overline{\mb{F}}_1^\dagger\mb{O}^{-1}\widetilde{\mb{\Psi}}_{1,2} \right)\left(-\widetilde{\mb{\Psi}}_{2,2}-\widetilde{\mb{\Psi}}_{2,1}\mb{O}^{-1}\widetilde{\mb{\Psi}}_{1,2}-\mb{\Omega}^{-1}\right)^{-1}\left(\overline{\mb{F}}_2+\widetilde{\mb{\Psi}}_{2,1}\mb{O}^{-1}\overline{\mb{F}}_1 \right)  \right)^{-1},\\
&\mc{G}_{4}(z)=\left( \mb{\Omega}- \left(-\widetilde{\mb{\Psi}}_{2,2}- \left(\overline{\mb{F}}_2+\widetilde{\mb{\Psi}}_{2,1}\mb{O}^{-1}\overline{\mb{F}}_1 \right)\left( -\mho-\overline{\mb{F}}_1^\dagger\mb{O}^{-1}\overline{\mb{F}}_1 \right)^{-1} \left(\overline{\mb{F}}_2^\dagger+\overline{\mb{F}}_1^\dagger\mb{O}^{-1}\widetilde{\mb{\Psi}}_{1,2} \right) \right)^{-1} \right)^{-1},\\
&\mc{G}_{5}(z)=\! \left(\!-\mb{\Phi}_{2}-\mb{I}\!- \! \overline{\mb{G}}_2\left(\!-\widetilde{\mb{\Phi}}_{2} -\left(\!-\widetilde{\mb{\Psi}}_{2,2}- \left(\!\overline{\mb{F}}_2+\widetilde{\mb{\Psi}}_{2,1}\mb{O}^{-1}\overline{\mb{F}}_1 \right)\left(\! -\mho-\!\overline{\mb{F}}_1^\dagger\mb{O}^{-1}\overline{\mb{F}}_1 \right)^{-1} \left(\overline{\mb{F}}_2^\dagger+\!\overline{\mb{F}}_1^\dagger\mb{O}^{-1}\widetilde{\mb{\Psi}}_{1,2} \!\right)\! \right)^{-1} \!\right)^{-1}\overline{\mb{G}}_2^\dagger\! \right)^{-1},\\
&\mc{G}_{6}(z)=\left(-\widetilde{\mb{\Psi}}_{2,2}- \left(\overline{\mb{F}}_2+\widetilde{\mb{\Psi}}_{2,1}\mb{O}^{-1}\overline{\mb{F}}_1 \right)\left( -\mho-\overline{\mb{F}}_1^\dagger\mb{O}^{-1}\overline{\mb{F}}_1 \right)^{-1} \left(\overline{\mb{F}}_2^\dagger+\overline{\mb{F}}_1^\dagger\mb{O}^{-1}\widetilde{\mb{\Psi}}_{1,2} \right) -\mb{\Omega}^{-1}\right)^{-1},\\
&\mc{G}_{7}(z)=\left( \mb{\Sigma}-\mb{A}_1\right)^{-1},\\
&\mc{G}_{8}(z)=\mc{G}_{2}(z)\overline{\mb{F}}_1\mb{A}_2-\mc{G}_{2}(z)\widetilde{\mb{\Psi}}_{1,2}\left(-\widetilde{\mb{\Psi}}_{2,2}+\overline{\mb{F}}_2\mb{\mho}^{-1}\overline{\mb{F}}_2^\dagger  \right)^{-1},\\
&\mc{G}_{9}(z)=\mc{G}_{3}(z)\left(\overline{\mb{F}}_2^\dagger+\overline{\mb{F}}_1^\dagger\mb{O}^{-1}\widetilde{\mb{\Psi}}_{1,2} \right)\left(-\widetilde{\mb{\Psi}}_{2,2}-\widetilde{\mb{\Psi}}_{2,1}\mb{O}^{-1}\widetilde{\mb{\Psi}}_{1,2}-\mb{\Omega}^{-1}\right)^{-1}. \label{G9}
		\end{align}
	}
	\hrulefill
\end{figure*}
in which the matrices $\mb{A}_2$-$\mb{A}_5$ and  the matrix functions $ \widetilde{\mb{\Psi}} $, $	{\mb{\Omega}}$, $	{\mb{\Pi}}$, $	\mb{\Sigma}$, and $	\mb{\Psi}_{i,j}$, $	\widetilde{\mb{\Psi}}_{i,j}$, $	\widetilde{\mb{\Phi}}_{i}$ and $	\mb{\Phi}_{i}$ are given in Appendix~\ref{prop1}.
\end{proposition}}


\begin{proof}
	The proof of Proposition~\ref{prop_cauchy1} is given in Appendix~\ref{appx_prop_cauchy}.
	
\end{proof}

With Proposition 1 and the relationship between the Cauchy transform $ \mc{G}_\mb{B}(z) $ and the Cauchy transform $ \mc{G}_\mu(z) $ in \eqref{Gu}, 
we obtain the information rate for two users with \eqref{R1_int} and \eqref{R2_int} by the numerical integration method. For the case $ R_2 < T < R_1+R_2 $, the Cauchy transforms for the two users are the same as in Proposition~\ref{prop_cauchy1}. The proof is similar to the case $T \leq R_2$ and therefore we omit it.
\subsection{Asymptotic Expression for Power Normalization Factor}
According to \cite{chenTWC2019}, the  power normalization factor for the GSVD precoder can be reformulated as 
\begin{align}
t=\text{Tr}\left(\mb{H}_1^\dagger\mb{H}_1+\mb{H}_2^\dagger\mb{H}_2\right)^{-1}.
\end{align}
Therefore, based on the operator-valued free probability theory and Anderson’s linearization trick for the matrix polynomials\cite{Belinschi}, we can obtain the asymptotic expression of the power normalization factor in the following proposition:
\begin{proposition}\label{prop_t}
	Define $\overline{\mb{G}}=\textup{diag}\{\overline{\mb{G}}_1,\overline{\mb{G}}_2\}$ and $\overline{\mb{F}}=[\overline{\mb{F}}_{1}^\dagger,\overline{\mb{F}}_{2}^\dagger]^\dagger$, 
	with \eqref{infty} holding and  $z \to 0$, the asymptotic expression of the power normalization factor  is given by 
	\begin{align}
	t=-\textup{Tr}(\hat{\mc{G}}_1),  \label{R2}
	\end{align}
	where 	$ \hat{\mc{G}}_1 $ satisfies the following matrix-valued equations
	\begin{align}
&	\hat{\mc{G}}_{1}(z) = \left(\widetilde{\mb{\Psi}}_{\textup{t}} \!- \!\overline{\mb{F}}^\dagger \left(-\widetilde{\mb{\Psi}}_{\textup{t}}-\!\left(-\widetilde{\mb{\Phi}}_{\textup{t}} - \! \overline{\mb{G}}^\dagger \mb{\Phi}_\textup{t}^{-1} \overline{\mb{G}}\right)^{-1}\right)^{-1}\overline{\mb{F}} \right)^{-1},\label{eqGCt}\\
	&	\hat{\mc{G}}_{2}(z)
=		 \left(-\widetilde{\mb{\Phi}}_{\textup{t}} - \overline{\mb{G}}^\dagger\mb{\Phi}_{\textup{t}}^{-1}\overline{\mb{G}} + \left(\widetilde{\mb{\Psi}}_{\textup{t}} + \overline{\mb{F}} \widetilde{\mb{\Psi}}_{\textup{t}}^{-1}\overline{\mb{F}}^\dagger\right)^{-1}\right)^{-1} ,\label{eqGDk} \\
&	\hat{	\mc{G}}_{3}(z) = \left(\mb{\Phi}_{\textup{t}} - \overline{\mb{G}} \left(\widetilde{\mb{\Phi}}_{\textup{t}} - \left(\widetilde{\mb{\Psi}}_{\textup{t}} + \overline{\mb{F}}\widetilde{\mb{\Psi}}_{\textup{t}}^{-1}\overline{\mb{F}}^\dagger \right)^{-1}\right)^{-1} \overline{\mb{G}}^\dagger \right)^{-1},\label{eqGDt}\\
&	\hat{	\mc{G}}_{4}(z)	= \left(-\widetilde{\mb{\Psi}}_{\textup{t}}- \overline{\mb{F}} \widetilde{\mb{\Psi}}_{\textup{t}}^{-1}\overline{\mb{F}}^\dagger + \left( \widetilde{\mb{\Phi}}_{\textup{t}} + \overline{\mb{G}}^\dagger\mb{\Phi}_{\textup{t}}^{-1}\overline{\mb{G}} \right)^{-1} \right)^{-1},\label{eqGCk}
\end{align}
where the matrices $	\widetilde{\mb{\Psi}}_\textup{t}$, $	\mb{\Phi}_\textup{t}$, $	\widetilde{\mb{\Phi}}_{\textup{t}}$, and $	\widetilde{\mb{\Psi}}_{\textup{t}}$ are given by 
\begin{align}
&\quad \quad \quad\quad  \widetilde{\mb{\Psi}}_\textup{t}=z\mb{I}-\mb{\Psi}, \\
&\quad \quad \quad \quad \mb{\Phi}_\textup{t}=\textup{diag}\left\{\sum_{k=1}^{K}\widetilde{\eta}_{k,1}(\hat{\mc{G}}_{2}),\sum_{k=1}^{K}\widetilde{\eta}_{k,2}(\hat{\mc{G}}_{2}) \right\}, \\
&\quad \quad \quad \quad \widetilde{\mb{\Phi}}_{\textup{t}}=\textup{diag}\{\widetilde{\mb{\Phi}}_{1}(	\hat{	\mc{G}}_{3}),\widetilde{\mb{\Phi}}_{2}(	\hat{	\mc{G}}_{3}) \},  \\
& \quad \quad \quad \quad \widetilde{\mb{\Psi}}_{\textup{t}}=
\left(\! {\begin{array}{*{20}{c}}
		{\widetilde{\mb{\Psi}}_{1,1}(\hat{\mc{G}}_{1})}&{\widetilde{\mb{\Psi}}_{1,2}(\hat{\mc{G}}_{1})}\\
		{\widetilde{\mb{\Psi}}_{2,1}(\hat{\mc{G}}_{1})}&{\widetilde{\mb{\Psi}}_{2,2}(\hat{\mc{G}}_{1})}
\end{array}}\!  \right),\\
&\mb{\Psi}_{\textup{t}}=\mb{\Psi}_{1,1}(\hat{	\mc{G}}_{4,1})+\mb{\Psi}_{1,2}(\hat{	\mc{G}}_{4,2})+\mb{\Psi}_{2,1}(\hat{	\mc{G}}_{4,3})+\mb{\Psi}_{2,2}(\hat{	\mc{G}}_{4,4}).
\end{align}

\end{proposition}
\begin{proof}
	The proof of Proposition~\ref{prop_t} is given in Appendix~\ref{appx_prop_t}.
\end{proof}

\section{Closed-Form expressions of the Asymptotic Information Rate and the Proposed PGAM Optimization Algorithm}
\subsection{Closed-Form Expressions of the Asymptotic Information Rate}
To obtain a deep insight into the considered multi-STAR-RISs assisted MIMO-NOMA system with GSVD-precoding, we consider a special
case by setting the channels between the BS and the STAR-RISs deterministic, i.e., 
	$ \mb{F}_{k}  = \overline{\mb{F}}_{k}$.  With the simplified model, the matrix functions $\mc{G}_1(z)$-$\mc{G}_7(z)$ in  Proposition~\ref{prop_cauchy1} can be simplified as 
		\begin{align}
			&	\mc{G}_{1}(z)	=\left(\widetilde{\mb{\Psi}}- \overline{\mb{G}}_1\left( -\widetilde{\mb{\Phi}}_{1,1}- \mb{A}_1^{-1}\right)\overline{\mb{G}}_1^\dagger \right), \label{G1_simply}	\\
						& \mb{A}_1	=	\left(-\widetilde{\mb{\Psi}}_{1,1}-\overline{\mb{F}}_1\left(-\mho -\overline{\mb{F}}_2^\dagger\mb{\Pi}^{-1}\overline{\mb{F}}_2 \right)^{-1}\overline{\mb{F}}_1^\dagger    \right)^{-1},	
				\end{align}	
				\begin{align}
			&\mc{G}_{2}(z)=\left( -\widetilde{\mb{\Psi}}_{1,1}-\mb{\Sigma}^{-1}-\overline{\mb{F}}_1\mb{A}_2\overline{\mb{F}}_1^\dagger  \right)^{-1},\\
			&\mc{G}_{3}(z)=\left( -\mho-\overline{\mb{F}}_1^\dagger\mb{O}^{-1}\overline{\mb{F}}_1-\overline{\mb{F}}_2^\dagger \mb{\Pi}^{-1}\overline{\mb{F}}_2   \right)^{-1},\\
			&\mc{G}_{4}(z)=\left( \mb{\Omega}+\left(\widetilde{\mb{\Psi}}_{2,2}- \overline{\mb{F}}_2 \left( \mho+\overline{\mb{F}}_1^\dagger\mb{O}^{-1}\overline{\mb{F}}_1 \right)^{-1} \overline{\mb{F}}_2^\dagger \right)^{-1} \right)^{-1},\\
			&\mc{G}_{5}(z)=\left(-\mb{\Phi}_{2}-\mb{I}-  \overline{\mb{G}}_2\left(-\widetilde{\mb{\Phi}}_{2} -\left(-\widetilde{\mb{\Psi}}_{2,2}
			 \right.\right.\right. \notag \\
			& \quad  \left. \left. \left.- \overline{\mb{F}}_2 \left(-\mho-\overline{\mb{F}}_1^\dagger\mb{O}^{-1}\overline{\mb{F}}_1 \right)^{-1} \overline{\mb{F}}_2^\dagger  \right)^{-1} \right)^{-1}\overline{\mb{G}}_2^\dagger \right)^{-1},\\
			&\mc{G}_{6}(z)=\left(\mb{\Pi}- \overline{\mb{F}}_2 \left( -\mho-\overline{\mb{F}}_1^\dagger\mb{O}^{-1}\overline{\mb{F}}_1 \right)^{-1} \overline{\mb{F}}_2^\dagger \right)^{-1},\\
			&\mc{G}_{7}(z)=\left( \mb{\Sigma}-\mb{A}_1\right)^{-1}, \label{G7_simply}
		\end{align}
where the matrices $ \widetilde{\mb{\Psi}} $, $	{\mb{\Omega}}$, $	{\mb{\Pi}}$, $	\mb{\Sigma}$, and $	\mb{\Psi}_{i,j}$, $	\widetilde{\mb{\Psi}}_{i,j}$, $	\widetilde{\mb{\Phi}}_{i}$ and $	{\mb{\Phi}}_{i}$ are given in Proposition~\ref{prop_cauchy1} by setting $\mc{G}_{8}(z)$ and $\mc{G}_{9}(z)$ to zero matrices.  Then based on the relationship between  the Cauchy transformation  $\mc{G}_\mb{B}(z)$ and the information rate for two users with \eqref{R1_int} and \eqref{R2_int}, we can derive the closed-form expressions of the asymptotic information rate in the following proposition:
\begin{proposition}\label{prop_shano}
	With \eqref{infty} holding and 	$ \mb{F}_{k}  = \overline{\mb{F}}_{k}$, the asymptotic information rates  for the two users  are given by 
	\begin{align}
		\overline{I}_1=&\phi(-\frac{1}{1+\alpha_1})-\phi(-1)+R_1\log\left({1+\alpha_1}\right), \label{R1} \\
		\overline{I}_2=&\phi(-1-\alpha_2)-(R_1-S)\log(1+\alpha_2)-\phi(-1-\kappa_1\alpha_2) \notag \\
		&	+(R_1-S)\log(1+\kappa_1\alpha_2),  \label{R2_simply}
	\end{align}
	where $\alpha_i=\frac{\rho_i}{t\sigma_0^2}$ and $ \phi(z) $ is given by 
	\begin{align}
	&	\phi(z)=\log\det\left({\widetilde{\mb{\Psi}}(z)}\right)+\log\det\left({{\mb{\Sigma}}(z)}\right)+\log\det\left(-\mb{\Phi}_{2}-\mb{I}\right) \notag \\
	&	+ \!\log\det\left(	\mb{\Omega} \right)+\!	\log\det\left({\mb{\Pi}}\right)+\!
		\log\det\left({\mb{O}}\right)+ \!\log\det\left(	\mc{G}_3(z)^{-1} \right) \notag \\
		&+   \textup{Tr}\left({\mb{\Phi}}_{1} \mc{G}_{1}(z)\right)
	 +\textup{Tr}\left({\mb{\Phi}}_{2} \mc{G}_{5}(z)\right)+	 \textup{Tr}\left(\widetilde{{\mb{\Psi}}}_{1,1} \mc{G}_{2}(z)\right) \notag \\
	& +\textup{Tr}\left(\widetilde{{\mb{\Psi}}}_{2,2} \mc{G}_{6}(z)\right),
	\end{align}
	where the involved matrices  are given by \eqref{G1_simply}-\eqref{G7_simply} and   Proposition 1 with setting $\mb{\Psi}_{1,2}$, $\mb{\Psi}_{2,1}$, $\widetilde{\mb{\Psi}}_{1,2}$, $\widetilde{\mb{\Psi}}_{2,1}$, $\mc{G}_{8}(z)$ and $\mc{G}_{9}(z)$ to zero matrices.
\end{proposition}
\begin{proof}
	The proof of Proposition~\ref{prop_shano} is given in Appendix~\ref{appx_shano}.
\end{proof}

\subsection{Proposed PGAM Optimization Algorithm}
\begin{algorithm}[h!]\label{op_theta}
	\caption{PGAM Optimization Algorithm} 
	\begin{algorithmic}[1]
		\State \textbf{initialize}:	Set the convergence criterion $ \varepsilon $, $j=0$,  $ \beta_{k,1,l} = \beta_{k,2,l}=0.5 $, 
		and randomly generate the phase-shifting matrix 	$ \mb{\Theta}^{(0)}$. Calculate $\bar{I}_1^{(0)}+ \bar{I}_2^{(0)} $ based on \eqref{R1} and \eqref{R2_simply}. 
		\State	\textbf{repeat} 
		\State \quad Calculate the step size $ \boldsymbol \alpha^{(i)} $ according to \cite{AbsilOp}.
		\State \quad Calculate gradient $ \boldsymbol{\delta}^{(j)}_i $ based on \eqref{grad}.
		\State \quad Update $\boldsymbol \theta^{(i+1)}$ based on \eqref{uptheta} and \eqref{PG}.
		\State \quad Update power normalization factor $t$ based on Proposition~\ref{prop_t}.
		\State \quad $ j=j+1 $.
		\State  \textbf{until} 	$ |\bar{I}_1^{(j)}+ \bar{I}_2^{(j)}-\bar{I}_1^{(j-1)} -\bar{I}_2^{(j-1)}| <\varepsilon $.
		\State  \textbf{output}  Optimal phase-shifting matrix $\mb{\Theta}$.
	\end{algorithmic}
\end{algorithm}
Based on the derived closed-form expressions of the information rate for two users, we can optimize the phase shifts and transmission and reflection coefficients of the STAR-RIS elements to maximize the sum rate of the two users. The optimization problem can be formulated as
\begin{subequations}
	\begin{align}
		\left ({\textbf {P1} }\right)&\max \limits _{   {{{\bf{\Theta}}}}    }~~	\bar{I}_1(\mb{\Theta})+ \bar{I}_2(\mb{\Theta})   \label{P1}  
		\\&  \mathrm {s.t.} \quad   \mb{\Theta}_{k,i}=\text{diag}(\sqrt{\beta_{k,i,1}}e^{\theta_{k,i,1}},...,\sqrt{\beta_{k,i,L_k}}e^{\theta_{k,i,L_k}}), \notag \\
		&\quad \quad \quad \quad \quad \quad \quad1\le k \le K,1\le i \le 2, \label{theta_st} \\
		& \quad  \beta_{k,1,l} + \beta_{k,2,l}= 1, 1\leq k\leq K, 1\leq l\leq L_k, \label{beta} 
	\end{align}	
\end{subequations}
where $\mb{\Theta}\in \{\mb{\Theta}_{1,1},...,\mb{\Theta}_{K,2}\}$. 
It is noted that due to the constant-modulus constraint
of phase shifts of the STAR-RIS,  Problem (P1) is non-convex. To address this issue, we propose a PGAM algorithm to solve Problem (P1). Specifically, by applying gradient ascent method with the constraints for the phase shifts and transmission and reflection coefficients of the STAR-RIS taken into account,
the sum rate increases monotonically iteratively.

Defining $ {\boldsymbol{\theta }}_{i}={\mb{\Theta}}_{i}\mb{1} $ for $1\leq i\leq 2$, where $ \mb{1} \in \mbb{C}^{L \times 1} $ is an all-one vector and $ {\mb{\Theta}}_i=\text{diag}\{\mb{\Theta}_{1,i},\mb{\Theta}_{2,i},...,\mb{\Theta}_{K,i} \} $. Then at the $(j+1)$-th iteration in the PGAM algorithm, the elements in the STAR-RIS are updated by 
\begin{align}
	\boldsymbol{\theta}^{(j+1)}_i(l) =  \text{P}_\text{G} \left( {\boldsymbol{\theta}^{(j)}_i(l)+ \boldsymbol{\alpha}^{(j)}_i(l)\boldsymbol{\delta}^{(j)}_i(l)}\right), \label{uptheta}
\end{align}
where $ \boldsymbol{\alpha}_i^{(j)} $  is the step size at the $ j $-th iteration and  $ \boldsymbol{\delta}^{(j)}_i $ is the gradient of $ \bar{I}_1+ \bar{I}_2 $ at the $ j $-th iteration, which is given by \eqref{vbz_deri} in Appendix~\ref{appx_gradient}.    The notation $\mb{a}(l)$ means the $l$-th element of the vector $\mb{a}$ and the operator $\text{P}_\text{G} \left( \boldsymbol{\theta}_i(l)\right)$ guarantees the constraints for the phase shifts and transmission and reflection coefficients of the STAR-RIS, which
is defined as
\begin{align}
\text{P}_\text{G} \left( \boldsymbol{\theta}_i(l)\right)=\frac{ \boldsymbol{\theta}_i(l)}{\sqrt{ |\boldsymbol{\theta}_2(l)|^2+|\boldsymbol{\theta}_1(l)|^2}}, \label{PG}
\end{align}
where the operator $|\boldsymbol{\theta}_i(l) |$ denotes the module of $\boldsymbol{\theta}_i(l)$.
Therefore, based on \eqref{uptheta} and \eqref{PG}, we proposed a PGAM optimization algorithm to solve Problem (P1), which is summarized in Algorithm 1.

	{ \color{black} For the complexity analysis, it can be  noted that the complexity of the proposed PGAM algorithm  mainly comes the calculation of the matrices in (75)-(82). Assuming that the  number of iterations of the fixed point equation in (75)-(82) is $ I_{in} $ and the   number of iterations for  the PGAM is  $ I_{out} $, the complexity of the PGAM algorithm is $O\left(I_{out}I_{in}\left((R_1+L)^3+(R_2+L)^3+T^3\right)\right)$ }

\section{Numerical Results}
\begin{table}[b] 
	\centering
	\caption{ Parameter settings of system configurations}
	\normalsize  
	\begin{tabular}{cccccc}
		\toprule
		& $ R_1 $ &  $ R_2 $ &  $ T $ &  $ \rho_1 $  &  $ \rho_2 $\\
		\midrule
		Case 1 &  16 & 16 &  10&  5&  1\\
		Case 2   & 10   & 10 &  16&  5&  1\\
		\bottomrule
	\end{tabular}
\end{table}
In this section, the numerical simulations are conducted to verify the accuracy of the derived asymptotic expression of the information rates and the effectiveness of the proposed PGAM algorithm. In the simulation,  for the Rician fading settings, the deterministic propagation components are modeled as the links between two UPAs equipped at both the transmitter and the receiver \cite{OESP}. If not specified, the phase shifts of RIS elements are randomly generated between 0 and $ 2\pi $ and the transmission and reflection coefficients are set to $\beta_{k,1,l}=\beta_{k,2,l}=0.5$. The number of elements in STAR-RIS is 30 and the number of the STAR-RIS panels is set to $K=2$.  In addition, we assume that the power allocation coefficients $\kappa_i$ for two users are set to $\kappa_1=0.1 $ and $\kappa_2=0.9$. Statistical parameters of the channel in the simulations are generated randomly but fixed in each Monte Carlo simulation, such as the deterministic unitary matrices  $\mb{U}_{k}$, $\mb{V}_{k}$, $\mb{T}_{k,i}$ and $\mb{S}_{k,i}$, as well as the variance matrices $\mb{M}_{k}$ and $\mb{N}_{k,i}$ in \eqref{eqFk}-\eqref{eqGk-2}. 

In addition, in the simulation, we consider the scenario where the attenuation of the direct link is significant due to the absorption of scatterers in the environment. Therefore, we set the channel gain of the direct links to be the same as that of each RIS-reflected link \cite{MoustakasTWC2023}, i.e.,
\begin{align}
\text{Tr}\left(\mb{R}_{0,i}\mb{R}_{0,i}^\dagger\right)=\text{Tr}\left(\left(\mb{G}_{k,i}\mb{\Theta}_{k,i}\mb{F}_k\right)\left(\mb{G}_{k,i}\mb{\Theta}_{k,i}\mb{F}_k\right)^\dagger\right), 
\end{align}  for $1\leq i\leq 2$ and $1\leq k\leq K$.

  \begin{figure}[b]
 	\centerline{\includegraphics[width=0.9\columnwidth]{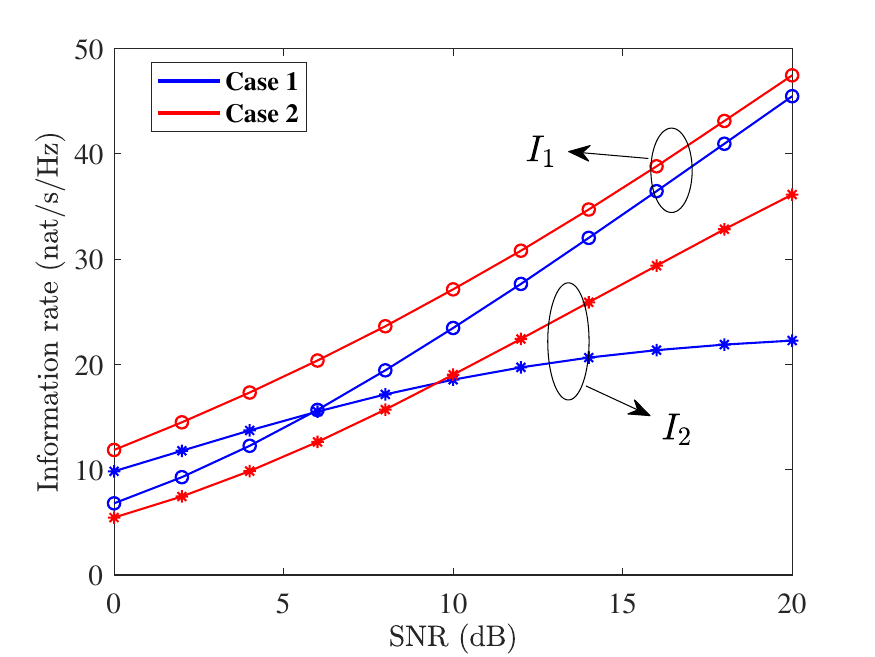}
 	}
 	\caption{{\color{black} Information rate under different system configurations. The solid lines represent the asymptotic results with Proposition~\ref{prop_cauchy1} and \eqref{R1_int} and \eqref{R2_int}, and the markers represent the Monte Carlo simulation results. The specific configuration of the two cases are shown in Table I.}} 
 	\label{sum1}
 \end{figure}
\begin{figure}[t]
	\centerline{\includegraphics[width=0.9\columnwidth]{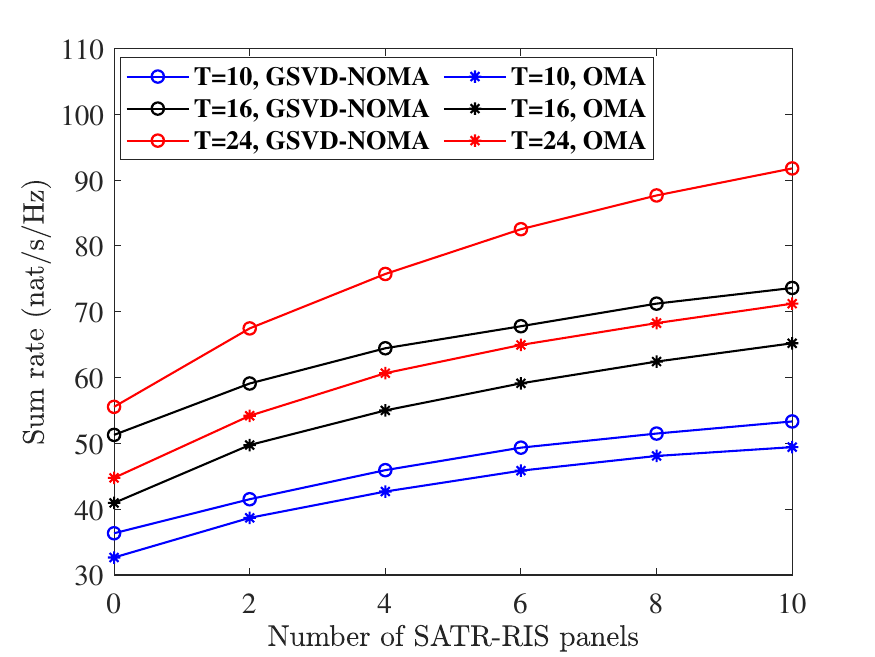}
	}
	\caption{{\color{black}Sum rate versus the number of STAR-RIS panels under different SNRs. The solid lines represent the asymptotic results with Proposition~\ref{prop_cauchy1} and \eqref{R1_int} and \eqref{R2_int}, and the markers represent the Monte Carlo simulation results.}} 
	\label{sumvsK}
\end{figure}
\begin{figure}[b]
	\centerline{\includegraphics[width=0.9\columnwidth]{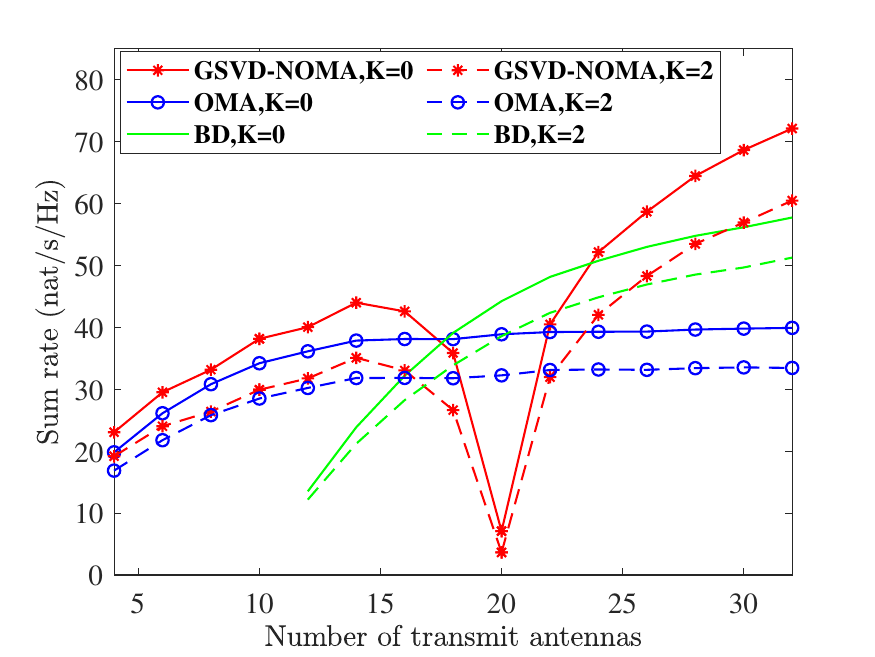}
	}
	\caption{{\color{black}Sum rate versus the number of transmit antennas under different transmission schemes. The solid lines and the markers in red and blue 
 represent the asymptotic results with Proposition~\ref{prop_cauchy1}, \eqref{R1_int} and \eqref{R2_int} and the Monte Carlo simulation results, respectively. The green lines represent the Monte Carlo simulation results for the BD scheme.}} 
	\label{sumvsT}
\end{figure}
\begin{figure}[t]
	\centerline{\includegraphics[width=0.9\columnwidth]{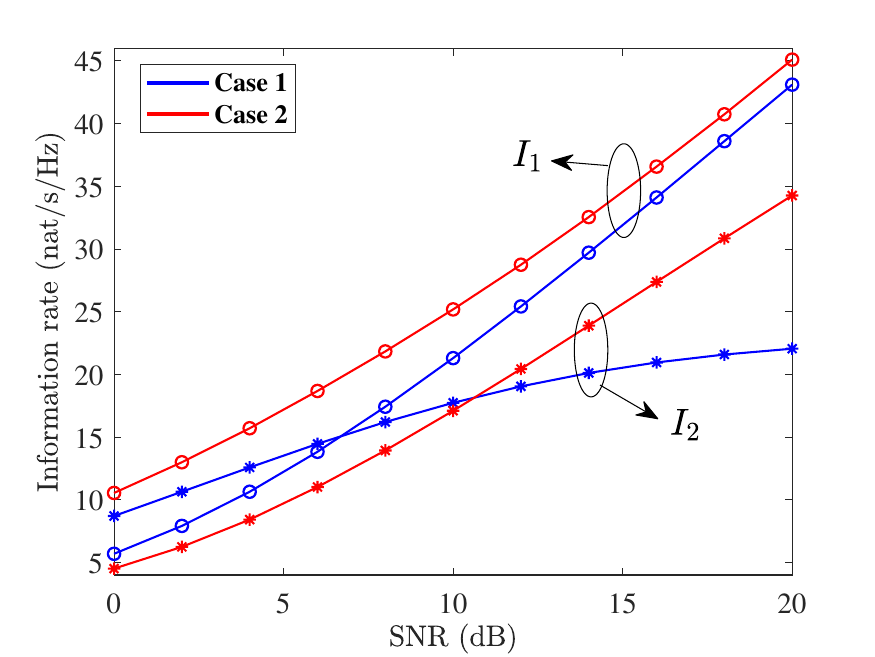}
	}
	\caption{{\color{black} Information rate versus the number of STAR-RIS panels under different SNRs. The solid lines represent the asymptotic results \eqref{R1} and \eqref{R2_simply}, and the markers represent the Monte Carlo simulation results. The specific configuration of the two cases are shown in Table I.}} 
	\label{simplify}
\end{figure}
\begin{figure}[b]
	\centerline{\includegraphics[width=0.9\columnwidth]{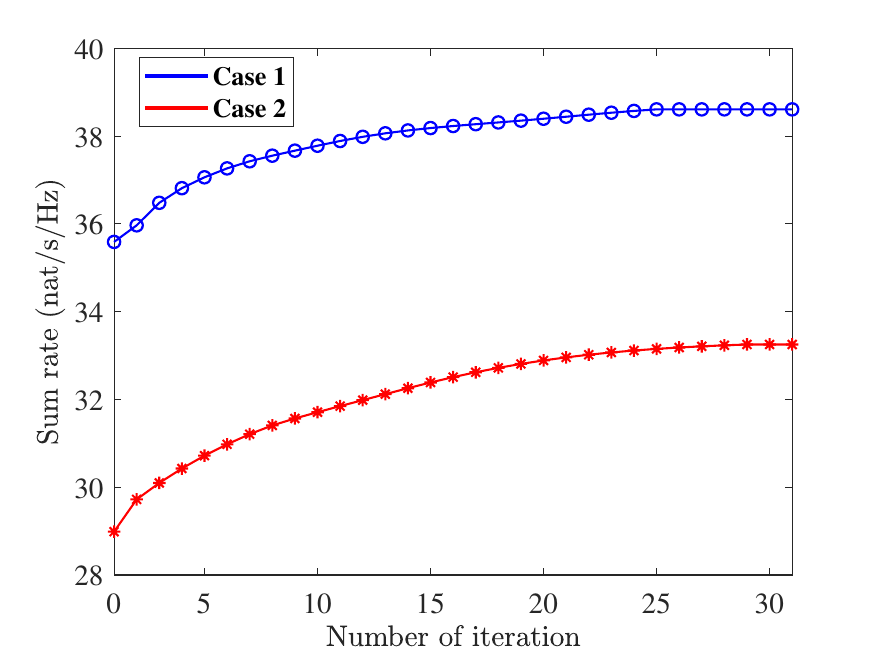}
	}
	\caption{{\color{black} Sum rate versus the number of iterations under different cases. The specific configuration of the two cases are shown in Table I. The solid lines represent the asymptotic results  \eqref{R1} and \eqref{R2_simply} and the markers represent the Monte Carlo simulation results.}} 
	\label{iter}
\end{figure}
\begin{figure}[t]
	\centerline{\includegraphics[width=0.9\columnwidth]{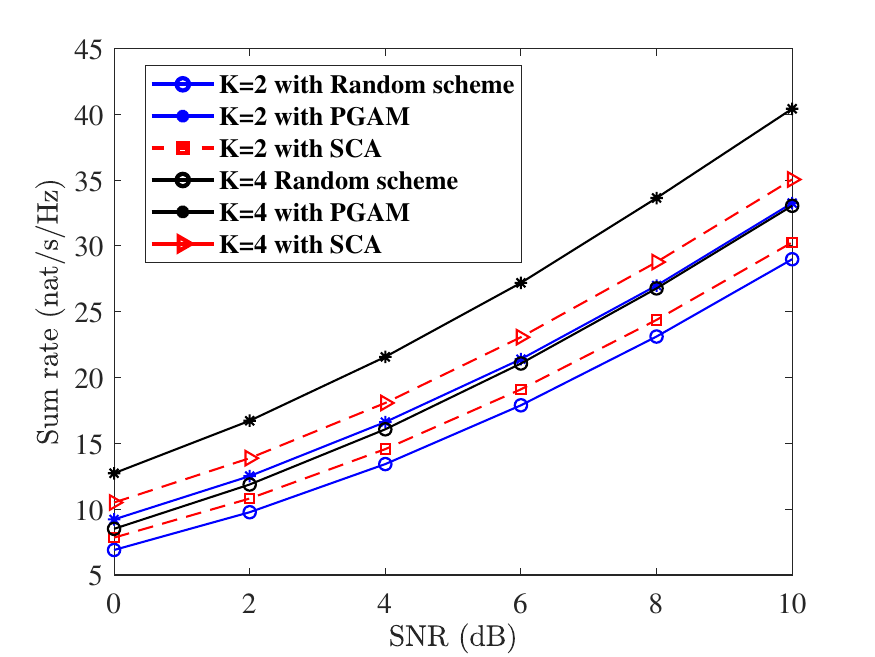}}
	\caption{{\color{black}Sum rate versus SNR under different  optimization algorithms. The solid lines and dashed lines represent the asymptotic results  \eqref{R1} and \eqref{R2_simply} and the markers represent the Monte Carlo simulation results.}}\label{op}
\end{figure}
To verify the accuracy of the derived asymptotic expressions, we plot  in Fig.~\ref{sum1} the information rates for two users versus signal-to-noise ratio (SNR) with different system configurations. 
As shown in Table I, we consider two cases with different system configurations to validate the derived results with Proposition~\ref{prop_cauchy1}. The SNR is defined as $\frac{1}{\sigma_0^2}$.
Each Monte Carlo simulation curve
in the figure is obtained by averaging over $ 10^6 $ channel realizations. In addition, for Case 2, the auxiliary variable $\Delta$ in \eqref{B} is set to $10^{-4}$ in the simulation. From Fig.~\ref{sum1}, we observe that the analytic results match the  Monte Carlo simulation results well, which shows the accuracy of the derived results of Proposition~\ref{prop_cauchy1}. In addition, in Fig.~\ref{sum1} of Case 1, as the SNR increases, the information rate of user 1 monotonically increases, but the information rate of user 2 tends to stabilize. This is due to the fact that with the SIC implemented at user 1, the interference of user 2 can be eliminated and hence the rate of user 1 continues to increase. However, the SIC is not deployed at user 2 and therefore, the information rate of user 2 is limited. For Case 2, the information rate for both users increases with the SNR. This is because the number of transmit antennas is larger than the number of receiving antennas in Case 2 and then there are orthogonal subchannels in the MIMO channel for transmission of user 2 in the GSVD-based precoding \cite{chenTWC2019,RAO2023}. Therefore, the information rate of user 2 is not limited by the interference of user 1 in Case 2.


To further demonstrate the merit of STAR-RIS in enhancing system performance, we depict 
Fig.~\ref{sumvsK} to show the relationship between the number of STAR-RIS and the sum rate $I_1+I_2$. The SNR is set to 20 dB and the number of receiving antennas is 16. In addition, we plot the sum rate under the OMA scheme to show the superiority of the GSVD-based NOMA scheme, in which {we consider the time division multiple access (TDMA) protocol and the transmission scheme in \cite{WangTWC2024}}. 
It can be observed from Fig.~\ref{sumvsK} that the analytic results in Proposition~\ref{prop_cauchy1} perfectly match the Monte Carlo simulation results for any number of STAR-RIS (including the no-RIS case, i.e., $K=0$), which further shows the generalization of our derived results. In addition, from Fig.~\ref{sumvsK}, {we can find that the sum rate increases with the number of STAR-RISs, but the growth of the sum rate slows down. This is due to the fact when the number of STAR-RIS is small, 
	the reflected links are sparse, then the increment of the number of STAR-RISs riches the paths of the channel, thereby  improving the sum rate effectively. However, excessive
 STAR-RIS panels  may cause path overlap, which cannot provide much performance gain.} Furthermore, compared to OMA scheme, the GSVD-based NOMA scheme can obtain a higher sum rate.

We plot Fig.~\ref{sumvsT} to further analyze the relationship between the sum rate of the GSVD precoding and the number of transmit antennas. In addition, the sum rate performance of OMA transmission scheme and the 
block diagonalization (BD) precoding scheme \cite{SpencerTSP2004} are included in Fig.~\ref{sumvsT}. The number of receiving antennas is set to 10 and the SNR is set to 10 dB. From Fig.~\ref{sumvsT}, we can observe that the sum rate of GSVD precoding  degrades when the number of transmit antennas approaches the sum of the number of receiving antennas $R_1+R_2$. This is because the power normalization factor of GSVD precoding increases significantly  with the number of transmit antennas when $ T \leq R_1+R_2$, which greatly decreases the power of the received signals thereby degrading the sum rate.  While when $ T > R_1+R_2$, the power normalization factor decreases with the increase of $ T $. At the same time, the increase in the number of transmit antennas leads to more spatial gain, which results in an increase in the sum rate. Therefore, we can conclude that when $ T \approx R_1+R_2 $, we should adopt other NOMA schemes or OMA schemes (e.g., the BD scheme or TDMA scheme)  to enhance the sum rate.

To demonstrate the accuracy of the derived closed-form expressions for information rate \eqref{R1} and \eqref{R2_simply}, we plot the information rate versus the SNR under different cases in Fig.~\ref{simplify}. To apply Proposition~\ref{prop_shano}, we only consider the fixed components in $\mb{F}_k$, i.e., $\widetilde{\mb{F}}_k=\mb{0}$. From Fig.~\ref{simplify}, it is observed that the curves of the asymptotic results follow the
Monte Carlo simulation accurately, which confirms the correctness of the derived closed-form expressions. In addition, the trends of the curves in Fig.~\ref{simplify} for $I_1$ and $I_2$ are the same as Fig.~\ref{sum1} since there are few changes in channel structure compared to Fig.~\ref{sum1}. 


To demonstrate the effectiveness of the proposed PGAM algorithm, we first depict the convergence of the PGAM algorithm (Algorithm 3) under different system cases in Fig.~\ref{iter}. The number of STAR-RIS panels is 2 and the number of elements for a STAR-RIS panel is 30. The SNR is fixed at 10 dB and the convergence criterion $ \varepsilon $ is set to $10^{-4}$. For both cases, the sum rate increases with the number of iterations and can achieve the maximum within 25 iterations. In addition, in Fig.~\ref{iter}, the asymptotic results  remain consistent with the simulation at any iteration, showing the compatibility of the closed-form expressions. Furthermore, the 0-th iteration in Fig.~\ref{iter} implies the initial value, which is much lower than the converged sum rate, thereby indicating the effectiveness of the proposed PGAM optimization algorithm. {\color{black} It should be noted that the converged performance of our proposed PGAM algorithm does not depend on the initialization.}

Finally, we plot Fig.~\ref{op} to demonstrate the superiority of the proposed optimization
algorithm with different number of STAR-RIS panels under Case 2. The number of RIS
panels is $K=2$ and $K=4$.   {\color{black} The random scheme is the case in which the phase shifts of the STAR-RISs are randomly set and the transmission and reflection coefficients are set to $\beta_1=\beta_2=0.5$. In addition, we compare the proposed PGAM algorithm with the successive convex approximation (SCA)-based algorithm \cite{Razaviyayn}, in which the results obtained by each iteration are projected into the appropriate space to satisfy the constraints of the STAR-RIS.  It can be clearly observed that the proposed optimization algorithm can effectively improve the sum rate in both cases. Furthermore, compared to the SCA-based algorithm, our proposed PGAM algorithm shows better performance, this is because the SCA is applied to solve
	the optimization problem without the constraints at each iteration 
	and  then projected the optimized phase shifts phase back to the appropriate space for the SATR-RIS. As a result, the iterations before and after cannot be matched directly. Therefore, the result of each iteration in the SCA process does not fully align with the previous iteration, which leads to the performance loss.}  Meanwhile, we can observe
that increasing the number of STAR-RIS panels increases the
optimization gain of the PGAM algorithm  when the number of system antennas is fixed. This is because that more STAR-RIS panels can provide a greater optimization space, and thus the performance gain provided by the PGAM algorithm becomes higher.

%

\section{Conclusion}
In this paper, we focus on the analysis of the information rate of the
MIMO-NOMA communication systems assisted by STAR-RIS, in which the GSVD precoding is applied at the BS. In addition, a general Rician channel model with Weichselberger's correlation structure is adopted to capture the correlation between the channel coefficients.  
Based on the operator-valued free probability theory, we obtain the Cauchy transform of the MIMO-NOMA channel for different antenna configurations and the closed-form expression for the power normalization factor. Then, by simplifying the channel, we 
derive the closed-form expression of the asymptotic information  rate for two users and propose a PGAM algorithm to enhance the sum rate by optimizing the STAR-RIS. 
 Additionally, the conducted numerical results reveal the
 accuracy and  the broadness  of the derived  results. Furthermore, the proposed optimization algorithm can further enhance the sum rate of the NOMA system. In addition, compared with the OMA scheme, the GSVD-based NOMA transmission scheme can obtain a higher sum rate when there is a gap between the  number of transmit antennas and the sum of the number of the receiving antennas.
\ifCLASSOPTIONcaptionsoff
  \newpage
\fi

	 	 
\begin{appendices}
		\section{ The Matrix-Valued Equations in Proposition~\ref{prop_cauchy1} }\label{prop1}
 {\color{black}
 	The matrices $\mb{A}_2$-$\mb{A}_5$ in Proposition~\ref{prop_cauchy1} are give by
		\begin{align}
			&\mb{A}_2	= \left(-\mho -\overline{\mb{F}}_2^\dagger\mb{\Pi}^{-1}\overline{\mb{F}}_2 \right)^{-1}, \\
			&\mb{A}_3=\left(-\widetilde{\mb{\Psi}}_{2,2}+\overline{\mb{F}}_2\mb{\mho}^{-1}\overline{\mb{F}}_2^\dagger-\mb{\Omega}^{-1}  \right)^{-1}\mb{\Omega}^{-1},\\
			&\mb{A}_4=\mb{\Omega}^{-1}\left(-\widetilde{\mb{\Psi}}_{2,2}+\overline{\mb{F}}_2\mb{\mho}^{-1}\overline{\mb{F}}_2^\dagger\right)^{-1}, \\
			&\mb{A}_5=\left(-\widetilde{\mb{\Psi}}_{2,2}+\overline{\mb{F}}_2\mb{\mho}^{-1}\overline{\mb{F}}_2^\dagger-\mb{\Omega}^{-1}  \right)^{-1}.
		\end{align}
	
	In addition, the matrix functions $ \widetilde{\mb{\Psi}} $, $	{\mb{\Omega}}$, $	{\mb{\Pi}}$, $	\mb{\Sigma}$, and $	\mb{\Psi}_{i,j}$, $	\widetilde{\mb{\Psi}}_{i,j}$, $	\widetilde{\mb{\Phi}}_{i}$ and $	\mb{\Phi}_{i}$ are respectively denoted as
	\begin{align}
		&	\widetilde{\mb{\Psi}} =z\mb{I}-\mb{\Phi}_{1}, \label{psi} \quad
		\mb{\Omega}= -	\widetilde{\mb{\Phi}}_2-\overline{\mb{G}}_2^\dagger\left(-{\Phi}_{2}-\mb{I}\right)^{-1}\overline{\mb{G}}_2,\\
		&	{\mb{\Pi}} = 	\widetilde{\mb{\Psi}}_{2,2}-	 \mb{\Omega}^{-1}, \quad
		{\mb{\Sigma}}= -\widetilde{\mb{\Phi}}_{1}-\overline{\mb{G}}_1^\dagger\widetilde{\mb{\Psi}}^{-1}\overline{\mb{G}}_1,\label{eqPsi2}\\
		&	\mb{O} = 	-\widetilde{\mb{\Psi}}_{1,1}-{\mb{\Sigma}}^{-1}, \mho =\mb{\Psi}_{1,1}+\mb{\Psi}_{2,2}-\mb{\Psi}_{1,2}-\mb{\Psi}_{2,1} ,\label{eqPsit2} \\
		&		\mb{\Phi}_1 = \sum_{k=1}^K \widetilde{\eta}_{k,1}(\mc{G}_{7}(z)), \quad \mb{\Phi}_2 = \sum_{k=1}^K \widetilde{\eta}_{k,2}(\mc{G}_{4}(z)), \\
		&	\widetilde{	\mb{\Phi}}_1= \mathrm{diag}\left\{ \mb{0},\  -\eta_{1,1}(\mc{G}_{1}(z)),\  \ldots,\  -\eta_{K,1}(\mc{G}_{1}(z))\right\},\\
		&		\widetilde{	\mb{\Phi}}_2= \mathrm{diag}\left\{ \mb{0},\  -\eta_{1,2}(\mc{G}_{5}(z)),\  \ldots,\  -\eta_{K,2}(\mc{G}_{5}(z))\right\},\\
		&		\widetilde{\mb{\Psi}}_{i,j} = \mathrm{diag}\left\{  \widetilde{\zeta}_{0,i}(\mc{G}_{3}(z)),  \widetilde{\zeta}_1(\mc{G}_{3}(z)), \ldots,    \widetilde{\zeta}_K(\mc{G}_{3}(z) \right\},\label{eqPsit}\\
		&				{\mb{\Psi}}_{1,1} ={\zeta}_{0,1}(\mc{G}_{2}(z))+ \sum_{k=1}^{K}{\zeta}_K(\mc{G}_{2}(z)),\\	
		&				{\mb{\Psi}}_{2,2} ={\zeta}_{0,1}(\mc{G}_{6}(z))+ \sum_{k=1}^{K}{\zeta}_K(\mc{G}_{6}(z)),	\\
		&{\mb{\Psi}}_{1,2} = \sum_{k=1}^{K}{\zeta}_k(\mc{G}_{9}(z)),	{\mb{\Psi}}_{1,2} = \sum_{k=1}^{K}{\zeta}_k(\mc{G}_{8}(z)),  \label{psi12}
	\end{align}	
	where the matrices $ \widetilde{\eta}_{k,i} $, $ \eta_{k,i} $, $ \widetilde{\zeta}_{k,i} $ and $ \zeta_{k,i} $ are the  parameterized one-sided correlation matrices of the channel as  defined in \eqref{yta}-\eqref{zeta_qta_k-n}.}
	\section{Proof of Proposition~\ref{prop_cauchy1}}\label{appx_prop_cauchy}
{\color{black}	According to \eqref{eqGB_GL}, we need to calculate    the operator-valued Cauchy transform    $ \mc{G}_{\mb{L}}^{\mc{D}}(\mb{\Lambda}(z)) $     to obtain the Cauchy transform  $ \mc{G}_{\mb{B}}(\mb{\Lambda}(z)) $. 	To explicitly calculate   $ \mc{G}_{\mb{L}}^{\mc{D}}(\mb{\Lambda}(z)) $, the linearization matrix $ \mb{L} $ is further expressed as 
	\begin{align}
		\mb{L} = \overline{\mb{L}} + \widetilde{\mb{L}},
	\end{align}
	where $\overline{\mb{L}}$ and $\widetilde{\mb{L}}$ contain the deterministic and the random parts of $\mb{L}$, respectively, and are given as follows:
	\begin{align}
		\setlength{\arraycolsep}{1.8pt}
		\overline{\mb{L}}=\left[ {\begin{array}{*{20}{c}}
				\mb{0}&	\mb{0}&	\mb{0}&	\mb{0}&	\mb{0}&	\mb{0}&	\overline{\mb{G}}_1\\
				\mb{0}&	\mb{0}&	\overline{\mb{F}}_1&	\mb{0}&	\mb{0}&	\mb{0}&-\mb{I}_{R_1+L}\\
				\mb{0}&	\overline{\mb{F}}_1^\dagger&	\mb{0}&	\mb{0}& \mb{0}&-\overline{\mb{F}}_2^\dagger&\mb{0}\\
				\mb{0}&	\mb{0}&	\mb{0}&	\mb{0}&-\overline{\mb{ G}}_2^\dagger&\mb{I}_{R_2+L}&	\mb{0}\\
				\mb{0}&	\mb{0}&	\mb{0}&-\overline{\mb{G}}_2&\mb{I}_{R_2}&	\mb{0}&	\mb{0}\\
				\mb{0}&	\mb{0}&-\overline{\mb{F}}_2&	\mb{I}_{R_2+L}&\mb{0}&\mb{0}&	\mb{0}\\
				\overline{\mb{G}}_1^\dagger&	-\mb{I}_{R_1+L}&	\mb{0}&	\mb{0}&\mb{0}&\mb{0}&	\mb{0}\\
		\end{array}} \right],
	\end{align}
	\begin{align}
		\setlength{\arraycolsep}{1.8pt}
		\widetilde{\mb{L}}=\left[ {\begin{array}{*{20}{c}}
				\mb{0}&	\mb{0}&	\mb{0}&	\mb{0}&	\mb{0}&	\mb{0}&	\widetilde{\mb{G}}_1\\
				\mb{0}&	\mb{0}&	\widetilde{\mb{F}}_1&	\mb{0}&	\mb{0}&	\mb{0}&\mb{0}\\
				\mb{0}&	\widetilde{\mb{F}}_1^\dagger&	\mb{0}&	\mb{0}& \mb{0}&-\widetilde{\mb{F}}_2^\dagger&\mb{0}\\
				\mb{0}&	\mb{0}&	\mb{0}&	\mb{0}&-\widetilde{\mb{G}}_2^\dagger&\mb{0}&	\mb{0}\\
				\mb{0}&	\mb{0}&	\mb{0}&-\mb{G}_2&\mb{0}&	\mb{0}&	\mb{0}\\
				\mb{0}&	\mb{0}&-\widetilde{\mb{F}}_2&	\mb{0}&\mb{0}&\mb{0}&	\mb{0}\\
				\widetilde{\mb{G}}_1^\dagger&	\mb{0}&	\mb{0}&	\mb{0}&\mb{0}&\mb{0}&	\mb{0}\\
		\end{array}} \right],
	\end{align}
	where $ \overline{\mb{G}}_i $,$ \overline{\mb{F}}_i $, $ \widetilde{\mb{G}}_i $ and $ \widetilde{\mb{F}}_i $ are defined as 
	\begin{align}
		\overline{\mb{G}}_i&=[\mb{I}_R, \overline{\mb{R}}_{1,i}\mb{\Theta}_{1,i},...,\overline{\mb{G}}_{K,i}\mb{\Theta}_{K,i}],\\
		\overline{\mb{F}}_i&=[\overline{\mb{F}}_{0,i}^\dagger, \overline{\mb{F}}_{1}^\dagger,...,\overline{\mb{F}}_{K}^\dagger]^\dagger,\\
		\widetilde{\mb{G}}_i&=[\mb{0}_R, \widetilde{\mb{R}}_{1,i}\mb{\Theta}_{1,i},...,\widetilde{\mb{G}}_{K,i}\mb{\Theta}_{K,i}],\\
		\widetilde{\mb{F}}_i&=[\widetilde{\mb{F}}_{0,i}^\dagger, \widetilde{\mb{F}}_{1}^\dagger,...,\widetilde{\mb{F}}_{K}^\dagger]^\dagger.
	\end{align}
	
	Then, with the similar steps as in \cite{ZhengTCom2023}, we can obtain that 
	$\overline{\mb{L}}$ is an operator-valued semicircular variable and is free from the deterministic matrix $\widetilde{\mb{L}}$. Therefore,  the operator-valued Cauchy transform    $ \mc{G}_{\mb{L}}^{\mc{D}}(\mb{\Lambda}(z)) $ can be further expressed with the subordination formula
	\cite{Belinschi} as 
	\begin{align}
		\mc{G}_{\mb{L}}^{\mc{D}}(\mb{\Lambda}(z)) &= \mc{G}_{\overline{\mb{L}}}^{\mc{D}}\left(\mb{\Lambda}(z) - \mc{R}_{\widetilde{\mb{L}}}^{\mc{D}}\left(\mc{G}_{\mb{L}}^{\mc{D}}(\mb{\Lambda}(z))\right)\right)\nonumber\\
		&= \mbb{E}_{\mc{D}}\left[\left(\mb{\Lambda}(z) - \mc{R}_{\widetilde{\mb{L}}}^{\mc{D}}\left(\mc{G}_{\mb{L}}^{\mc{D}}(\mb{\Lambda}(z))\right) - \overline{\mb{L}}\right)^{-1}\right],\label{eqGL_sub}
	\end{align}
	where $\mc{R}_{\widetilde{\mb{L}}}^{\mc{D}}\left(\cdot\right)$ denotes the operator-valued $R$-transform of $\mb{L}$ over $\mc{D}$. Then we need to calculate the operator-valued $R$-transform 	$ \mc{R}_{\widetilde{\mb{L}}}^{\mc{D}}\left(\mc{G}_{\mb{L}}^{\mc{D}}(\mb{\Lambda}(z))\right) $. 
	According to \cite{ZhengTCom2023}, given a matrix $\mb{K}\in    \mbb{C}^{n \times n}$,  the operator-valued $R$-transform 	$ \mc{R}_{\widetilde{\mb{L}}}^{\mc{D}}\left(\mc{G}_{\mb{L}}^{\mc{D}}(\mb{\Lambda}(z))\right) $ can be expressed as 
	\begin{align}
		\mc{R}_{\widetilde{\mb{L}}}^{\mc{D}}(\mb{K}) &= \mbb{E}_{\mc{D}}\left[\widetilde{\mb{L}}\mb{K}\widetilde{\mb{L}}\right]\nonumber\\
		&=\left[\begin{array}{c:c:c:c:c:c:c}
			\mb{\Phi}_1 & & & & & &\\\hdashline
			&-\widetilde{\bm{\Psi}}_{1,1} & & & &-\widetilde{\bm{\Psi}}_{1,2}& \\\hdashline
			&&{\mb{\mho}} & & &&\\ \hdashline
			& & &  \widetilde{\bm{\Phi}}_2& && \\ \hdashline 
			& &  & &	 {\bm{\Phi}}_1 &&\\ \hdashline
			& -\widetilde{\bm{\Psi}}_{2,1}& & & & -\widetilde{\bm{\Psi}}_{2,2} &\\ \hdashline
			& & & & && \mb{\Phi}_2
		\end{array}\right],\label{eqEDX1}
	\end{align}
	where the matrix functions $ \widetilde{\mb{\Psi}} $, $	\mb{\Psi}_{i,j}$, $	\widetilde{\mb{\Psi}}_{i,j}$, $	\widetilde{\mb{\Phi}}_{i}$ and $	\mb{\Phi}_{i}$ are shown as \eqref{psi}-\eqref{psi12} in Appendix~\ref{prop1} and the matrix  $\mb{K}$ has the same structure as $ \mbb{E}_\mc{D}\left[\mb{X}\right] $	in \eqref{eqEDX}.
		
	Therefore, by replacing $ \mb{K} $ in \eqref{eqEDX1} with $\mc{G}_{\mb{L}}^{\mc{D}}(\mb{\Lambda}(z)) $, and substituting
	$ \overline{\mb{L}} $ and $ \mc{R}_{\widetilde{\mb{L}}}^{\mc{D}}\left(\mc{G}_{\mb{L}}^{\mc{D}}(\mb{\Lambda}(z))\right) $ into \eqref{eqGL_sub}, respectively, we obtain
	$ \mc{G}_{\mb{L}}^{\mc{D}}(\mb{\Lambda}(z)) $ as
				\setlength{\arraycolsep}{1.8pt}
	\begin{align}
		&\mc{G}_{\mb{L}}^{\mc{D}}(\mb{\Lambda}(z)) \\
		&= \left[\begin{array}{c:c:c:c:c:c:c}
			z\mb{I}-\mb{\Phi}_1 & & & & & &-\overline{\mb{G}}_1\\\hdashline
			&-\widetilde{\bm{\Psi}}_{1,1} &-\overline{\mb{F}}_1 & & &-\widetilde{\bm{\Psi}}_{1,2}&\mb{I} \\\hdashline
			& -\overline{\mb{F}}_1^\dagger&{\mb{\mho}} & & &-\overline{\mb{F}}_2^\dagger&\\ \hdashline
			& & &  \widetilde{\bm{\Phi}}_2&\overline{\mb{G}}_2^\dagger &-\mb{I}& \\ \hdashline 
			& &  &\overline{\mb{G}}_2 &	 -{\bm{\Phi}}_1-\mb{I} &&\\ \hdashline
			& \widetilde{\bm{\Psi}}_{2,1}&\overline{\mb{F}}_2 &-\mb{I} & & -\widetilde{\bm{\Psi}}_{2,2} &\\ \hdashline
			-\overline{\mb{G}}_1^\dagger&\mb{I} & & & && -\mb{\Phi}_2
		\end{array}\right]^{-1}.\label{eqEG}
	\end{align}
	
	Then by continuously applying the block matrix inversion identity   to \eqref{eqEG}, we can obtain the expression of the  operator-valued Cauchy transform    $ \mc{G}_{\mb{L}}^{\mc{D}}(\mb{\Lambda}(z)) $ and then conclude the Cauchy transform  $ \mc{G}_{\mb{B}}(\mb{\Lambda}(z)) $ with $ 	\mc{G}_\mb{B}(z) = \frac{1}{R_1}\mathrm{Tr}\left(\left\{\mc{G}_{\mb{L}}^{\mc{D}}(\mb{\Lambda}(z))\right\}^{(1,1)}\right) $, which are finally summarized in Proposition 1.
	\section{Proof of Proposition~\ref{prop_t}}\label{appx_prop_t}
	The  power normalization factor for the GSVD precoder can be further reformulated as 
	\begin{align}
		t\mathop  = 
		\limits^{z \to 0}-\text{Tr}\left(z\mb{I}-\mb{H}_1^\dagger\mb{H}_1-\mb{H}_2^\dagger\mb{H}_2\right),
	\end{align}
which is equal to the 	Cauchy transform of $\mb{B}_\textup{t}=\mb{H}_1^\dagger\mb{H}_1+\mb{H}_2^\dagger\mb{H}_2$. Therefore,  we aim to obtain the	Cauchy transform of $ \mb{B}_\textup{t} $.

Define ${\mb{G}}=\textup{diag}\{{\mb{G}}_1,{\mb{G}}_2\}$ and ${\mb{F}}=[{\mb{F}}_{1}^\dagger,{\mb{F}}_{2}^\dagger]^\dagger$, the matrix $ \mb{B}_\textup{t} $ can be expressed as $ \mb{B}_\textup{t}={\mb{F}}^\dagger{\mb{G}}^\dagger{\mb{G}}{\mb{F}} $. Then following the similar method of Section III-A, we first construct the linearized matrix $\mb{L}_{\textup{t}}$ as
\begin{align}
	\mb{L}_{\textup{t}}  = \left[\begin{array}{cccc}
		\mb{0} & \mb{0} & \mb{0} & \mb{F}^\dagger\\
		\mb{0} & \mb{0} & \mb{G}^\dagger& -\mb{I}\\
		\mb{0} & \mb{G} & -\mb{I}& \mb{0}\\
	\mb{F} 	 & -\mb{I} & \mb{0} & \mb{0}
	\end{array}\right].\label{eqBL}
\end{align}
	Then substituting $ \bf{G} $ and $  \bf{F}  $  into \cite[Prop. 2]{ZhengTCom2023} and following the same steps,  we can obtain the Cauchy transform of $ \mb{B}_\textup{t} $. The detailed proof  is an application of Proposition 2 in \cite[Prop. 2]{ZhengTCom2023} and therefore, is omitted here. Finally, with Cauchy transform of $ \mb{B}_\textup{t} $ and $ z \to 0 $, we can obtain Proposition~\ref{prop_t}.}

	\section{Derivation of $\overline{I}_1$ and $\overline{I}_2$}\label{appx_shano}

	Based on \eqref{R1_int} and \eqref{R2_int}, by differentiating $I_1$ and $I_2$ with respect to $ \alpha_1 $ and $\alpha_2$, we have 
	\begin{align}
		\frac{\partial I_1}{\partial \alpha_1}&= \frac{S}{1+\alpha_1} + \frac{S}{(1+\alpha_1)^{2}}\mc{G}_\mu(-(1+\alpha_1)^{-1}), \\
		\frac{\partial I_1}{\partial \alpha_2}&= -S\mc{G}_\mu(-(1+\alpha_2))+S\mc{G}_\mu(-(1+\kappa_1\alpha_2)),
	\end{align}
where $ \alpha_1 $ and $\alpha_2$ are defined in Proposition~\ref{prop_shano}.	Then to prove the correctness of Proposition~\ref{prop_shano},  we only need to prove that $	\frac{\partial I_1}{\partial \alpha_1}=	\frac{\partial \overline{I}_1}{\partial \alpha_1}$ and $	\frac{\partial \overline{I}_2}{\partial \alpha_2}=	\frac{\partial \overline{I}_2}{\partial \alpha_2}$ hold with the given expression in Proposition~\ref{prop_shano} and as $\sigma_0^2$ approaches infinity, $ \overline{I}_1 $ and $ \overline{I}_2 $ are equal to 0.  

 Therefore,  the partial derivative of $ \overline{I}_1 $ with respect to $ \alpha_1 $ can be expressed as 
	\begin{align}
	&		\frac{\partial\overline{I}_1(\alpha_1)	}{\partial\alpha_1}=\frac{\partial\phi(\frac{-1}{1+\alpha_1})}{\partial\alpha_1} + \frac{R_1}{1+\alpha_1} \notag \\
&	=\!\frac{\partial\! \log\det\left({\widetilde{\mb{\Psi}}(\lambda_1 )}\right)	}{\partial\alpha_1}\!+\!\!  \frac{\partial\! \log\det\left({{\mb{\Sigma}}(\lambda_1)}\right)	}{\partial\alpha_1}\!+\!\! \frac{\partial\! \log\det\left(-\mb{\Phi}_{2}\!-\!\mb{I}\right)	}{\partial\alpha_1}
 \notag \\
	&+\frac{\partial \log\det \left(	\mb{\Omega} \right)	}{\partial\alpha_1}	+ \frac{\partial \log\det\left({\mb{\Pi}}\right)	}{\partial\alpha_1}	+  \frac{\partial \log\det\left(	\mc{G}_3(\lambda_1)^{-1} \right)	}{\partial\alpha_1}
 \notag \\
	&+\frac{\partial 	\log\det\left({\mb{O}}\right)	}{\partial\alpha_1}+ \frac{\partial \textup{Tr}\left({\mb{\Phi}}_{1} \mc{G}_{1}(\lambda_1)\right)	}{\partial\alpha_1}  + \frac{\partial \textup{Tr}\left({\mb{\Phi}}_{2} \mc{G}_{5}(\lambda_1)\right)	}{\partial\alpha_1} \notag \\
	& +\frac{\partial  \textup{Tr}\left(\widetilde{{\mb{\Psi}}}_{1,1} \mc{G}_{2}(\lambda_1)\right)	}{\partial\alpha_1}  	+\frac{\partial  \textup{Tr}\left(\widetilde{{\mb{\Psi}}}_{2,2} \mc{G}_{6}(\lambda_1)\right)	}{\partial\alpha_1} + \frac{R_1}{1+\alpha_1},
\end{align}
where $\lambda_1=\frac{-1}{1+\alpha_1}$.
 For a matrix function $\mb{K}(z)$, we have
\begin{align}
	\frac{\partial}{\partial z} \log \det \mb{K}(z)&= \mathrm{Tr} \left(\mb{K}(z)^{-1}\frac{\partial \mb{K}(z)}{\partial z} \right), \label{maxtr_derive}
	\\	\mathrm{Tr}\left(\frac{\partial \mb{K}(z)^{-1}}{\partial z}    \right)	&= -\mathrm{Tr} \left(\mb{K}(z)^{-1}\frac{\partial \mb{K}(z)}{dz}\mb{K}(z)^{-1} \right). \label{maxtr_derive2}
\end{align}
Then we can obtain 
\begin{align}
&\frac{\partial \log\det\left({\widetilde{\mb{\Psi}}(\lambda_1 )}\right)	}{\partial\alpha_1}=\text{Tr}\left(\widetilde{\mb{\Psi}}(\lambda_1 )^{-1}\frac{\partial {\widetilde{\mb{\Psi}}(\lambda_1 )}	}{\partial\alpha_1} \right), \label{Psi_t}\\
&\frac{\partial \log\det\left({{\mb{\Sigma}}(\lambda_1 )}\right)	}{\partial\alpha_1}=\text{Tr}\left({\mb{\Sigma}}(\lambda_1 )^{-1}\frac{\partial {{\mb{\Sigma}}(\lambda_1 )}	}{\partial\alpha_1} \right),\\
&\frac{\partial \log\det\left(-\mb{\Phi}_{2}-\mb{I}\right)	}{\partial\alpha_1}=\text{Tr}\left(\left(-\mb{\Phi}_{2}-\mb{I}\right)^{-1}\frac{\partial \left(-\mb{\Phi}_{2}-\mb{I}\right)	}{\partial\alpha_1} \right),\\
&\frac{\partial \log\det \left(	\mb{\Omega} \right)	}{\partial\alpha_1}	=\text{Tr}\left(\mb{\Omega}^{-1}\frac{\partial \mb{\Omega}	}{\partial\alpha_1} \right),\\
&\frac{\partial \log\det\left({\mb{\Pi}}\right)	}{\partial\alpha_1}=\text{Tr}\left({\mb{\Pi}}^{-1}\frac{\partial {\mb{\Pi}}	}{\partial\alpha_1} \right),\\
&\frac{\partial 	\log\det\left({\mb{O}}\right)	}{\partial\alpha_1}	=\text{Tr}\left({\mb{O}}^{-1}\frac{\partial {\mb{O}}	}{\partial\alpha_1} \right), \\
&\frac{\partial \log\det\left(	\mc{G}_3(\lambda_1)^{-1} \right)	}{\partial\alpha_1}
=\text{Tr}\left(\mc{G}_3(\lambda_1)\frac{\partial \mc{G}_3(\lambda_1)^{-1}	}{\partial\alpha_1} \right). \label{G3}
\end{align}

 In addition, 	according to Lemma 1 in \cite{WangTWC2024} about the  parameterized one-sided correlation matrices in \eqref{yta} to \eqref{zeta_qta_k-n}, we can obtain
 \begin{align}
&\frac{\partial \textup{Tr}\left({\mb{\Phi}}_{1} \mc{G}_{1}(\lambda_1)\right)	}{\partial\alpha_1}=\textup{Tr}\left( \mc{G}_{1}(\lambda_1)\frac{\partial  {\mb{\Phi}}_{1}	}{\partial\alpha_1}\right)+ \textup{Tr}\left(  {\mb{\Phi}}_{1} \frac{\partial\mc{G}_{1}(\lambda_1)	}{\partial\alpha_1}\right) \notag \\
&=\textup{Tr}\left( \mc{G}_{1}(\lambda_1)\frac{\partial  \left(\lambda_1\mb{I}- \widetilde{\mb{\Psi}}\right)	}{\partial\alpha_1}\right) +\textup{Tr}\left( \mc{G}_{7}(\lambda_1)\frac{\partial  \widetilde{\mb{\Phi}}_{1,1}	}{\partial\alpha_1}\right) \notag \\
&=-\lambda_1^2-\textup{Tr}\left( \mc{G}_{1}(\lambda_1)\frac{\partial  \widetilde{{\mb{\Psi}}}	}{\partial\alpha_1}\right)+\textup{Tr}\left( \mc{G}_{7}(\lambda_1)\frac{\partial  \widetilde{\mb{\Phi}}_{1,1}	}{\partial\alpha_1}\right).\label{G1fai11}
 \end{align}
Applying  the matrix inversion lemma to $ \mc{G}_{1}(\lambda_1) $ and \eqref{maxtr_derive2}, we can obtain 
\begin{align}
&\textup{Tr}\left( \mc{G}_{1}(\lambda_1)\frac{\partial  \widetilde{{\mb{\Psi}}}	}{\partial\alpha_1}\right) \notag \\
&=\textup{Tr}\left( \left( \widetilde{{\mb{\Psi}}}^{-1} +
\widetilde{{\mb{\Psi}}}^{-1}\overline{\mb{G}}_1\mc{G}_{7}(\lambda_1)\overline{\mb{G}}_1^\dagger\widetilde{{\mb{\Psi}}}^{-1}
\right)
\frac{\partial  \widetilde{{\mb{\Psi}}}	}{\partial\alpha_1}\right) \notag \\
&=\text{Tr}\left(\widetilde{\mb{\Psi}}^{-1}\frac{\partial {\widetilde{\mb{\Psi}}}	}{\partial\alpha_1} \right)-\textup{Tr}\left( \mc{G}_{7}(\lambda_1)\overline{\mb{G}}_1^\dagger\frac{\partial  \widetilde{{\mb{\Psi}}}^{-1}	}{\partial\alpha_1}\overline{\mb{G}}_1\right). \label{G7fai11}
\end{align}
Combining  \eqref{G1fai11} and \eqref{G7fai11} and applying  the matrix inversion lemma to $ \mc{G}_{7}(\lambda_1) $, we can rewrite $ \frac{\partial \textup{Tr}\left({\mb{\Phi}}_{1} \mc{G}_{1}(\lambda_1)\right)	}{\partial\alpha_1} $ as 
\begin{align}
\frac{\partial \textup{Tr}\left({\mb{\Phi}}_{1} \mc{G}_{1}(\lambda_1)\right)	}{\partial\alpha_1}=&-\lambda_1^2-\text{Tr}\left(\widetilde{\mb{\Psi}}^{-1}\frac{\partial {\widetilde{\mb{\Psi}}}	}{\partial\alpha_1} \right)-\text{Tr}\left({\mb{\Sigma}}^{-1}\frac{\partial {{\mb{\Sigma}}}	}{\partial\alpha_1} \right)\notag \\
&
+\textup{Tr}\left( \mc{G}_{2}\frac{\partial {{\mb{\Sigma}}}^{-1}	}{\partial\alpha_1}\right).
\end{align}

Similarly, we can  express $ \frac{\partial \textup{Tr}\left({\mb{\Phi}}_{2} \mc{G}_{5}(\lambda_1)\right)	}{\partial\alpha_1} $ as 
\begin{align}
&\frac{\partial \textup{Tr}\left({\mb{\Phi}}_{2} \mc{G}_{5}(\lambda_1)\right)	}{\partial\alpha_1}=\textup{Tr}\left( \mc{G}_{5}(\lambda_1)\frac{\partial  {\mb{\Phi}}_{2}	}{\partial\alpha_1}\right)+ \textup{Tr}\left(  {\mb{\Phi}}_{2} \frac{\partial\mc{G}_{5}(\lambda_1)	}{\partial\alpha_1}\right) \notag \\
&=\textup{Tr}\left( \mc{G}_{5}(\lambda_1)\frac{\partial  {\mb{\Phi}}_{2}	}{\partial\alpha_1}\right)+ \textup{Tr}\left( \mc{G}_{4}(\lambda_1)\frac{\partial  \widetilde{{\mb{\Phi}}}_{2}	}{\partial\alpha_1}\right). \label{fai2g5}
\end{align}
With the matrix inversion lemma applied to $ \mc{G}_{5}(\lambda_1) $, we can obtain 
\begin{align}
&\textup{Tr}\left( \mc{G}_{5}(\lambda_1)\frac{\partial  {\mb{\Phi}}_{2}	}{\partial\alpha_1}\right)=-\textup{Tr}\left( \left( \left(-\mb{\Phi}_{2}-\mb{I}\right)^{-1} \right. \right. \notag \\
&\left.\left. + \left(-\mb{\Phi}_{2}-\mb{I}\right)^{-1}\overline{\mb{G}}_2  \mc{G}_{4}\overline{\mb{G}}_2^\dagger\left(-\mb{\Phi}_{2}-\mb{I}\right)^{-1} \right)\frac{\partial \left(- {\mb{\Phi}}_{2}-\mb{I}\right)	}{\partial\alpha_1}\right) \notag \\
&=-\textup{Tr}\left( \left(-\mb{\Phi}_{2}-\mb{I}\right)^{-1}  \frac{\partial \left(- {\mb{\Phi}}_{2}-\mb{I}\right)	}{\partial\alpha_1}\right) \notag \\
&+\textup{Tr}\left(  \mc{G}_{4}\overline{\mb{G}}_2^\dagger \frac{\partial \left(-\mb{\Phi}_{2}-\mb{I}\right)^{-1}	}{\partial\alpha_1}\overline{\mb{G}}_2 \right). \label{G5fai2}
\end{align}
Substituting  \eqref{G5fai2} into \eqref{fai2g5} and with the matrix inversion lemma applied to $ \mc{G}_{4}(\lambda_1) $, we can rewrite $\frac{\partial \textup{Tr}\left({\mb{\Phi}}_{2} \mc{G}_{5}(\lambda_1)\right)	}{\partial\alpha_1} $ as 
\begin{align}
&\frac{\partial \textup{Tr}\left({\mb{\Phi}}_{2} \mc{G}_{5}(\lambda_1)\right)	}{\partial\alpha_1}
=
-\textup{Tr}\left( \left(-\mb{\Phi}_{2,2}-\mb{I}\right)^{-1}  \frac{\partial \left(- {\mb{\Phi}}_{2}-\mb{I}\right)	}{\partial\alpha_1}\right) \notag \\
&- \textup{Tr}\left(  \mb{\Omega} \frac{\partial\mb{\Omega}	}{\partial\alpha_1} \right)+\textup{Tr}\left(  \mc{G}_6 \frac{\partial\mb{\Omega}	}{\partial\alpha_1} \right) .
\end{align}

For the terms $ \frac{\partial  \textup{Tr}\left(\widetilde{{\mb{\Psi}}}_{1,1} \mc{G}_{2}(\lambda_1)\right)	}{\partial\alpha_1} $ and $ \frac{\partial  \textup{Tr}\left(\widetilde{{\mb{\Psi}}}_{2,2} \mc{G}_{6}(\lambda_1)\right)	}{\partial\alpha_1} $, we have
\begin{align}
	&\frac{\partial  \textup{Tr}\left(\widetilde{{\mb{\Psi}}}_{1,1} \mc{G}_{2}(\lambda_1)\right)	}{\partial\alpha_1}
\notag \\
&	=\textup{Tr}\left( \mc{G}_{2}(\lambda_1)\frac{\partial  \widetilde{{\mb{\Psi}}}_{1,1}	}{\partial\alpha_1}\right)+ \textup{Tr}\left(  \widetilde{{\mb{\Psi}}}_{1,1} \frac{\partial\mc{G}_{2}(\lambda_1)	}{\partial\alpha_1}\right) \notag \\
	&=\textup{Tr}\left( \mc{G}_{2}(\lambda_1)\frac{\partial  \widetilde{{\mb{\Psi}}}_{1,1}	}{\partial\alpha_1}\right)+ \textup{Tr}\left( \mc{G}_{3}(\lambda_1)\frac{\partial  {{\mb{\Psi}}}_{1,1}	}{\partial\alpha_1}\right), \\
		&\frac{\partial  \textup{Tr}\left(\widetilde{{\mb{\Psi}}}_{2,2} \mc{G}_{6}(\lambda_1)\right)	}{\partial\alpha_1}
		 \notag \\
		&=\textup{Tr}\left( \mc{G}_{6}(\lambda_1)\frac{\partial  \widetilde{{\mb{\Psi}}}_{2,2}	}{\partial\alpha_1}\right)+ \textup{Tr}\left(  \widetilde{{\mb{\Psi}}}_{2,2} \frac{\partial\mc{G}_{6}(\lambda_1)	}{\partial\alpha_1}\right) \notag \\
	&=\textup{Tr}\left( \mc{G}_{6}(\lambda_1)\frac{\partial  \widetilde{{\mb{\Psi}}}_{2,2}	}{\partial\alpha_1}\right)+ \textup{Tr}\left( \mc{G}_{3}(\lambda_1)\frac{\partial  {{\mb{\Psi}}}_{2,2}	}{\partial\alpha_1}\right).
\end{align}
 With the matrix inversion lemma applied to $ \mc{G}_{6}(\lambda_1) $ and $ \mc{G}_{2}(\lambda_1) $, we can obtain 
 \begin{align}
   &\mc{G}_{6}(\lambda_1)=\mb{\Pi}^{-1}+\mb{\Pi}^{-1}\overline{\mb{F}}_2 \mc{G}_{3}(\lambda_1)\overline{\mb{F}}_2^\dagger\mb{\Pi}^{-1},  \label{G6inv}\\
 &  \mc{G}_{2}(\lambda_1)=\mb{O}^{-1}+\mb{O}^{-1}\overline{\mb{F}}_1 \mc{G}_{3}(\lambda_1)\overline{\mb{F}}_1^\dagger\mb{O}^{-1}. \label{G2inv}
 \end{align}
 Finally, combining \eqref{Psi_t}-\eqref{G3} and the other simplified terms as well as \eqref{G6inv} and \eqref{G2inv} and with some simple math manipulation, we can simplify $ \frac{\partial\overline{I}_1(\alpha_1)	}{\partial\alpha_1} $ as
	\begin{align}
	\frac{\partial\overline{I}_1(\alpha_1)	}{\partial\alpha_1}=	&		-\lambda_1^2\text{Tr}\left(	\mc{G}_{\widetilde{\mb{C}}}\left(\lambda_1\right)\right)+ \frac{R_1}{1+\alpha_1} \notag \\
		&=\frac{R_1}{(1+\alpha_1)^2}\mc{G}_{\mb{B}}(-(1+\alpha_1)^{-1}) + \frac{R_1}{1+\alpha_1}.
	\label{V_derive2}
	\end{align}

By substituting \eqref{B} into \eqref{V_derive2}, we can derive that
\begin{align}
&\frac{\partial\overline{I}_1}{\partial\alpha_1}=\frac{R_1}{1+\alpha_1} \notag \\
	&+\frac{R_1}{(1+\alpha_1)^2}\left(  \frac{S}{R_1}\mc{G}_\mu(-(1+\alpha_1)^{-1}) + \frac{(S-R_1)(1+\alpha_1)}{R_1}\right)\notag \\
&	=\frac{S}{1+\alpha_1} + \frac{S}{(1+\alpha_1)^{2}}\mc{G}_\mu(-(1+\alpha_1)^{-1}) \notag \\
&=\frac{\partial{I}_1}{\partial\alpha_1}.
\end{align}

Furthermore, as $\sigma_0^2$ approaches infinity, $\alpha_1=0$ and in this case we have 
	\begin{align}
	\overline{I}_1 =\phi(-1)-\phi(-1)+R_1\log\left({1}\right)=0.
	\end{align}
Therefore, we verify the correctness of $\overline{I}_1$. Since $\overline{I}_2$ has a similar structure to $\overline{I}_1$, their proof process is similar and we omit the detailed derivation. As a result, we complete the proof of Proposition 3.

\section{Gradient of $ \bar{I}_1 $ + $ \bar{I}_2 $}\label{appx_gradient}
 
\begin{figure*}[t] 
	{
		\begin{align} 
			\frac{\partial 
				\phi(x)}{d\mb{\Theta}_{k,i,l}^*}=&\mathrm{Tr}\left(-\left(\hat{\mb{\Theta}}_i^\dagger\mb{\Gamma}_1\hat{\mb{\Theta}}_i\widetilde{\mb{\Psi}}_{i,i}-\mb{I}\right)^{-1}\mb{E}_{k,i,l}\mb{\Gamma}_1\hat{\mb{\Theta}}_i\widetilde{\mb{\Psi}}_{i,i} \right) +\mathrm{Tr}\left(\mc{G}_3(x)  \overline{\mb{F}}^\dagger\mb{\Gamma}_2\overline{\mb{F}} \right),  \label{grad}
			\\
			\mb{\Gamma}_1=&-	\widetilde{\mb{\Phi}}_i-\overline{\mb{R}}_i^\dagger\left(-{\Phi}_{i}-\mb{I}\right)^{-1}\overline{\mb{R}}_i,
	\quad
			\mb{\Gamma}_2=-\left(-\widetilde{\mb{\Psi}}_{i,i} - \left(\hat{\bf{\Theta}}_{i}^\dagger\mb{\Gamma}_1\hat{\bf{\Theta}}_i \right)^{-1}\right)^{-1}\mb{\Gamma}_3\left(-\widetilde{\mb{\Psi}}_{i,i} - \left(\hat{\bf{\Theta}}_{i}^\dagger\mb{\Gamma}_1\hat{\bf{\Theta}}_i \right)^{-1}\right)^{-1}, \label{gamma1} 
			\\
			\mb{\Gamma}_3 =&\left(\hat{\mb{\Theta}}_i^\dagger\mb{\Gamma}_1\hat{\mb{\Theta}}_i-\mb{I}\right)^{-1}\mb{E}_{k,i,l}\mb{\Gamma}_1\hat{\mb{\Theta}}_i\left(\hat{\mb{\Theta}}_i^\dagger\mb{\Gamma}_1\hat{\mb{\Theta}}_i-\mb{I}\right)^{-1}. \label{gamma2} 
		\end{align}
	}
	\hrulefill
\end{figure*}
According to Proposition~\ref{prop_shano}, we can obtain 
 \begin{align}
&\boldsymbol{\delta}_i=2\frac{\partial (\bar{I}_1+\bar{I}_2 )}{\boldsymbol{\theta}_i^*} \notag \\
&=2\frac{\partial 
	\left( \phi(\frac{-1}{1+\alpha_1})-\phi(-1)+\phi(-1-\alpha_2)-\phi(-1-\kappa_1\alpha_2) \right)}{\boldsymbol{\theta}_i^*}
	, \label{vbz_deri}
\end{align}
where $ \frac{\partial 
 \phi(x)}{\boldsymbol{\theta}_i^*}=\left[\frac{\partial 
 \phi(x)}{d\mb{\Theta}_{1,i,1}^*},
\frac{\partial 
	\phi(x)}{d\mb{\Theta}_{1,i,2}^*},...,\frac{\partial 
	\phi(x)}{d\mb{\Theta}_{K,i,L_k}^*}\right]^\text{T} $. The derivative $ \frac{\partial 
	\phi(x)}{d\mb{\Theta}_{k,i,l}^*} $  is given by \eqref{grad} at the top of this page,
where we omit the notation $ (x) $ for convenience. The matrix $\hat{\mb{\Theta}}_i$ is defined as $\hat{\mb{\Theta}}_i=\text{diag}\{\mb{I},\mb{\Theta}_i\}$ and $ \mb{E}_{k,i,l} $ is defined as $ \frac{\partial\hat{\mb{\Theta}}_i}{\partial [\mb{\Theta}_{k,i}]_{l,l}} $.
\end{appendices}

\end{document}